\documentclass[12pt]{amsart} 
\usepackage[left=3cm,top=2.5cm,bottom=2.5cm,right=3cm]{geometry}
\usepackage[oldstyle]{libertine}
\usepackage{amssymb, amsmath, amsthm}
\usepackage{aliascnt}
\usepackage{mathtools}
\usepackage{graphicx,xspace}
\usepackage{epsfig}
\usepackage{enumitem}
\usepackage[usenames,dvipsnames]{xcolor}
\usepackage{tikz}
\usepackage[T1]{fontenc}
\usepackage[utf8]{inputenc}
\makeatletter
\renewcommand*\libertine@figurestyle{LF}
\makeatother
\usepackage[libertine,libaltvw,liby]{newtxmath}
\makeatletter
\renewcommand*\libertine@figurestyle{OsF}
\usepackage{subcaption}
\definecolor{green}{RGB}{0,127,0}
\definecolor{red}{RGB}{191,0,0}
\usepackage[colorlinks,cite color=red,link color=green,pagebackref=true]{hyperref}
\usepackage{todonotes}
\usepackage[capitalize]{cleveref}
\usepackage{bm} % For bold math symbols
\usepackage{seqsplit}
\usepackage{comment}
\usetikzlibrary{matrix, arrows, calc, decorations.markings, patterns, shapes}
\theoremstyle{plain}
\newtheorem{lemma}{Lemma}[section]
\newaliascnt{theorem}{lemma}
\newtheorem{theorem}[theorem]{Theorem}
\aliascntresetthe{theorem}

\newaliascnt{theoremdefinition}{lemma}
\newtheorem{theoremdefinition}[theoremdefinition]{Theorem-Definition}
\aliascntresetthe{theoremdefinition}

\newaliascnt{corollary}{lemma}
\newtheorem{corollary}[corollary]{Corollary}
\aliascntresetthe{corollary}

\newaliascnt{proposition}{lemma}

\aliascntresetthe{proposition}

\newaliascnt{problem}{lemma}

\aliascntresetthe{problem}

\newaliascnt{question}{lemma}

\aliascntresetthe{question}

\newaliascnt{conjecture}{lemma}

\aliascntresetthe{conjecture}

\newaliascnt{definition}{lemma}
\newtheorem{definition}[definition]{Definition}
\aliascntresetthe{definition}

\newaliascnt{assumption}{lemma}

\aliascntresetthe{assumption}

\newaliascnt{claim}{lemma}

\aliascntresetthe{claim}

\theoremstyle{remark}
\newtheorem{remark}{Remark}
\newtheorem{example}{Example}

\crefname{theorem}{theorem}{theorems}
\Crefname{theorem}{Theorem}{Theorems}
\crefname{lemma}{lemma}{lemmas}
\Crefname{lemma}{Lemma}{Lemmas}
\crefname{corollary}{corollary}{corollaries}
\Crefname{corollary}{Corollary}{Corollaries}
\crefname{proposition}{proposition}{propositions}
\Crefname{proposition}{Proposition}{Propositions}
\crefname{problem}{problem}{problems}
\Crefname{problem}{Problem}{Problems}
\crefname{question}{question}{questions}
\Crefname{question}{Question}{Questions}
\crefname{conjecture}{conjecture}{conjectures}
\Crefname{conjecture}{Conjecture}{Conjectures}
\crefname{definition}{definition}{definitions}
\Crefname{definition}{Definition}{Definitions}
\crefname{assumption}{assumption}{assumptions}
\Crefname{assumption}{Assumption}{Assumptions}
\crefname{claim}{claim}{claims}
\Crefname{claim}{Claim}{Claims}
\crefname{theoremdefinition}{theorem-definition}{theorem-definitions}
\Crefname{theoremdefinition}{Theorem-Definition}{Theorem-Definitions}
\crefname{remark}{remark}{remarks}
\Crefname{remark}{Remark}{Remarks}
\crefname{example}{example}{examples}
\Crefname{example}{Example}{Examples}

\AtBeginDocument{\let\cref\Cref}

\DeclareMathOperator{\Sym}{Sym}
\DeclareMathOperator{\Alpha}{Alpha}
\DeclareMathOperator{\Beta}{Beta}
\DeclareMathOperator{\Planch}{Planch}
\DeclareMathOperator{\dd}{d\!}

\DeclareMathOperator{\sgn}{sgn}
\DeclareMathOperator{\spec}{spec}

\newcommand{\Lapformal}{\widetilde{\mathcal{L}}}
\newcommand{\Lap}{\mathcal{L}}
\newcommand{\C}{\mathbb{C}}
\newcommand{\R}{\mathbb{R}}
\newcommand{\Q}{\mathbb{Q}}
\newcommand{\Z}{\mathbb{Z}}
\newcommand{\N}{\mathbb{N}}

\newcommand{\M}{\mathcal{M}}

\newcommand{\PP}{\mathbb{P}}
\newcommand{\MM}{\mathcal{M}}
\newcommand{\E}{\mathbb{E}}
\newcommand{\Y}{\mathbb{Y}}
\newcommand{\LL}{\mathscr{L}}

\newcommand{\xx}{\mathbf{x}}
\newcommand{\yy}{\mathbf{y}}
\newcommand{\Luk}{\mathbf{L}}
\newcommand{\la}{\lambda}
\newcommand{\Prob}{\mathbb P}
\newcommand{\Sy}[1]{\mathfrak{S}_{#1}}

\title[Crystallization of discrete $N$-particle systems at high temperature]{Crystallization of discrete $N$-particle systems at high temperature}

\author[C.~Cuenca]{Cesar Cuenca}
\address{
Department of Mathematics,
The Ohio State University,
231 West 18th Avenue,
Columbus, OH 43210, USA.
}
\email{cesar.a.cuenk@gmail.com}

\author[M.~Dołęga]{Maciej Dołęga}
\address{
Institute of Mathematics, 
Polish Academy of Sciences, 
ul. Śniadeckich 8, 
00-956 Warszawa, Poland.
}
\email{mdolega@impan.pl}

\thanks{CC was partially supported by the NSF grant DMS-2348139 and by the Simons Foundation's Travel Support for Mathematicians grant MP-TSM-00006777. This research was funded in whole or in part by {\it Narodowe Centrum Nauki}, grant 2021/42/E/ST1/00162. For the purpose of Open Access, the author has applied a CC-BY public copyright licence to any Author Accepted Manuscript (AAM) version arising from this submission.}

\begin{document}
\emergencystretch 3em

\begin{abstract}
This is the second paper in a series studying the global asymptotics of discrete $N$-particle systems with inverse temperature parameter $\theta$ in the high temperature regime. In the first paper, we established necessary and sufficient conditions for the Law of Large Numbers at high temperature in terms of Jack generating functions. In this paper, we derive a functional equation for the moment generating function of the limiting measure, which enables its analysis using analytic tools. We apply this functional equation to compute the densities of the high temperature limits of the \emph{pure Jack measures}. As a special case, we obtain the high temperature limit of the large fixed-time distribution of the discrete-space $\beta$-Dyson Brownian motion of Gorin--Shkolnikov.

Two special cases of our densities are the high temperature limits of discrete versions of the G$\beta$E, computed in \cite{AllezBouchaudGuionnet2012}, and L$\beta$E, computed in \cite{AllezBouchaudMajumdarVivo2013}. Moreover, we prove the following crystallization phenomenon of the particles in the high temperature limit: the limiting measures are uniformly supported on disjoint intervals with unit gaps and their locations correspond to the zeros of explicit special functions with all roots located in the real line. We also show that these zeros correspond to the spectra of certain unbounded Jacobi operators.
\end{abstract}

\maketitle

\tableofcontents

\section{Introduction}

\subsection{Overview}

The study of random $N$-particle systems on the real line, inspired by models from Statistical Physics and Random Matrix Theory, has long been a central topic in probability. For specific values of the inverse temperature parameter $\beta$ --- notably $\beta=1,2,4$ --- these systems correspond to eigenvalue distributions of classical random matrices, and Wigner’s semicircle law~\cite{Wigner1958} is the foundational result for their large-$N$ behavior.

\smallskip

A remarkable feature is that the global limiting distribution of the aforementioned systems, known as $\beta$-ensembles or one-dimensional log-gas systems, is universal in the sense that, as the size of the system tends to infinity, it depends only on the potential and not on the inverse temperature $\beta$, provided that $\beta N\to \infty$~\cite{Forrester2010}. By contrast, in the high temperature regime, where $\beta$ decreases with $N$ so that $\beta N \to \gamma > 0$, the limiting law becomes a one-parameter family determined by $\gamma$. This transition, which bridges the gap between independent particles and classical ensembles, has been rigorously analyzed for several models, including the Hermite and Laguerre $\beta$-ensembles~\cite{AllezBouchaudGuionnet2012,AllezBouchaudMajumdarVivo2013,DuyShirai2015,TrinhTrinh2021}, as well as in more general settings~\cite{Benaych-GeorgesCuencaGorin2022,MergnyPotters2022}, where precise criteria for the Law of Large Numbers (LLN) in this regime have been established.

\smallskip

A similar phenomenon was expected to occur in the discrete counterparts of $\beta$-ensembles.
And indeed, it has been confirmed that, for various models, the global limiting distribution is universal in the fixed temperature regime~\cite{DolegaFeray2016,BorodinGorinGuionnet2017,DolegaSniady2019,Huang2021,Moll2023,CuencaDolegaMoll2023,DimitrovGaoGuNiedernhofer2025}. The high temperature regime has also been analyzed recently for three different models~\cite{DolegaSniady2019,CuencaDolegaMoll2023,CuencaDolega2025} through the powerful method of moments.
The proofs naturally led to various moment-cumulant type formulae, always expressed as weighted generating functions of certain lattice paths, known as \emph{Łukasiewicz paths}, which generalize well-known moment-cumulant type formulae from classical and free probability.

\smallskip

In our recent paper~\cite{CuencaDolega2025}, we established necessary and sufficient conditions for the LLN of discrete $N$-particle systems with inverse temperature parameter $\beta$, when the inverse temperature decays proportionally to the number of particles, in terms of Jack generating functions with parameter $\theta$, where $\beta = 2\theta$. The Jack generating functions can be informally interpreted as a version of the Fourier transform that is well-suited for the study of many concrete discrete models.
An important outcome of our moment method approach is the aforementioned explicit combinatorial description of the moments of the limiting measure, expressed in terms of the \emph{quantized $\gamma$-cumulants} that are limits of Taylor coefficients of the Jack generating functions. Moreover, we showed that the limiting measure is uniquely determined by the moments. Nevertheless, several natural questions remained open, e.g.: What do the limiting measures look like? Do they admit densities? Can we understand our combinatorial moment-cumulant formula in terms of an analytic framework?

\smallskip

This is the second paper in a series studying the global asymptotics of discrete $N$-particle systems with inverse temperature parameter $\beta$ in the high temperature regime, and its purpose is to address the problems raised above, among others. The last question is a natural path to answer the previous ones and also, potentially, to extend our analysis to models that do not have all moments finite. Inspired by the success of the analytic approach to Free Probability theory --- where, in that generality, instead of working with the sequence of cumulants, one works with the $R$-transform and proves its existence as an analytic function for a large class of probability measures by means of a functional equation between the $R$-transform and the Cauchy transform --- we develop a similar approach for the discrete $\beta$-ensembles in the high temperature regime. We prove a functional equation between the \emph{$\gamma$-quantized $R$-transform}, which is the generating function of the quantized $\gamma$-cumulants, and the moment generating function. Remarkably, this functional equation differs drastically from the combinatorial transform in~\cite{CuencaDolega2025} and from the analytic formula obtained also from the high temperature regime, but for continuous models~\cite{Benaych-GeorgesCuencaGorin2022}.
%Its proof is based on various analytic arguments and, surprisingly, an algebro-geometric argument about the rigidity of polynomial identities.

\smallskip

The functional equation discussed above enables us to analyze the limiting measures using analytic tools, and we apply it to compute the densities of the high temperature limits of the \emph{pure Jack measures}, which are discrete $N$-particle systems that depend on the Jack parameter $\theta>0$ and arise from Jack-positive specializations of the ring of symmetric functions.
As a special case, we obtain the high temperature limit of the large fixed-time distribution of the discrete-space $\beta$-Dyson Brownian motion of Gorin--Shkolnikov~\cite{GorinShkolnikov2015}; the parameters are related by $\beta=2\theta$.
%, which is a discrete version of the Dyson Brownian motion with an additional parameter $\theta>0$ that is related to $\beta$ by $\beta=2\theta$.
In addition, two special cases of our densities are discrete analogues of the high temperature limits of the G$\beta$E, computed in \cite{AllezBouchaudGuionnet2012}, and the L$\beta$E, computed in \cite{AllezBouchaudMajumdarVivo2013}.

\smallskip

Although the structure of the combinatorial moment formulas for the limiting measures in the continuous and discrete settings are similar, our simulations in~\cite{CuencaDolega2025} suggested that the densities in the discrete setting are qualitatively different; see~\cref{fig:Markov_Simulation}. Indeed, we prove that the following crystallization phenomenon of the particles occurs in the high temperature limit: the limiting measures are uniformly supported on disjoint intervals with gaps of unit length.
The right endpoints of these intervals turn out to be determined by the zeros of explicit special functions involving Gauss and confluent hypergeometric functions; see \cref{fig:Comparison1/2,fig:Comparison1}.
We will show that these functions are real-rooted by relating their zeros to the spectra of certain unbounded Jacobi operators.
It is worth mentioning that a similar crystallization phenomenon was previously observed in~\cite{DolegaSniady2019} for the limit of Jack--Plancherel measures in the high temperature regime, and it was proved in~\cite{CuencaDolegaMoll2023} that the locations of the particles in the limit correspond to the zeros of yet another special function, the Bessel function of the first kind.
However, the nature of the crystallization in that case was a direct consequence of the scaling of the model, while in the present situation, the fact that the particles are supported on disjoint intervals with unit gaps comes as a surprising consequence of the detailed analysis of the functional equation.

\smallskip

After this overview, we proceed to state our main results more precisely. We start by recalling the combinatorial transform from~\cite{CuencaDolega2025} that expresses the moments of the limiting measure in terms of the \emph{quantized $\gamma$-cumulants}. Then we state our main theorem, which is a functional equation between the $\gamma$-quantized $R$-transform and the moment generating function. Finally, we discuss some applications for computing densities of limiting measures for discrete $N$-particle systems in the high temperature regime. %One can also skip this more detailed description of our results and go directly to \cref{sec:jacks}.

\subsection{LLN for discrete N-particle systems at high temperature and the gamma-quantized R-transform}

Let $\{\PP_N\}_{N\ge 1}$ be a sequence, where each $\PP_N$ is a probability measure on the set $\Y(N)$ of integer partitions $\la_1\ge\cdots\ge\la_N\ge 0$, possibly depending on the Jack parameter $\theta > 0$.
We assume that all $\PP_N$ have small tails, so that their \emph{Jack generating functions}, defined by
\begin{equation}
G_{N,\theta}(x_1,\dots,x_N) := \sum_{\la\in\Y(N)}\PP_N(\la)\frac{P_\la(x_1,\dots,x_N;\theta)}{P_\la(1^N;\theta)},
\end{equation}
are analytic in disks of radii larger than $1$; above, $P_\la(x_1,\dots,x_N;\theta)$ are the Jack polynomials. In our previous paper~\cite{CuencaDolega2025} we gave a characterization of the LLN for discrete $N$-particle systems $\LL_1 > \cdots > \LL_N$, with $\LL_i := \la_i - \theta (i-1)$, in the \emph{high temperature regime}
\begin{equation}\label{eqn:high_intro}
N\to\infty,\quad \theta\to 0,\quad N\theta\to\gamma,
\end{equation}
where $\gamma\in(0,\infty)$ is a fixed constant. We proved that there exist $m_1,m_2,\dots\in\R$ such that, for all $s\in\Z_{\ge 1}$ and $k_1,\dots,k_s\in\Z_{\ge 1}$, we have the following convergence of mixed moments:
\begin{equation}
\lim_{\substack{N\to\infty,\\ N\theta \to \gamma}}\frac{1}{N^s}\,\E_{\PP_N} \left[ \prod_{j=1}^s{ \sum_{i=1}^N{\LL_i^{k_j}} } \right]
= \prod_{j=1}^s{m_{k_j}}
\end{equation}
if and only if the following two conditions are satisfied:
\begin{enumerate}
	\item $\displaystyle\lim_{\substack{N\to\infty\\N\theta\to\gamma}} \frac{1}{(n-1)!}\cdot\frac{\partial^n}{\partial x_1^n} \,\ln(G_{N,\theta})\,\Big|_{(x_1,\dots,x_N) = (1^N)} = \kappa_n^\gamma$ exists and is finite, for all $n \in \Z_{\geq 1}$,

	\item $\displaystyle\lim_{\substack{N\to\infty\\N\theta\to\gamma}} \frac{\partial^r}{\partial x_{i_1} \cdots \partial x_{i_r}} \,\ln(G_{N,\theta})\,\Big|_{(x_1,\dots,x_N) = (1^N)} = 0$, for all $r\ge 2$ and all $i_1,\dots,i_r\in\Z_{\geq 1}$ with at least two distinct indices among $i_1,\dots,i_r$.
\end{enumerate}
Consequently, if $m_1,m_2,\dots$ is the moment sequence of a unique probability measure $\mu^\gamma$, then $\mu^\gamma$ is the weak limit of the empirical measures $\mu^{(N)} := \frac{1}{N}\sum_{i=1}^N \delta_{\LL_i}$, in the regime~\eqref{eqn:high_intro}.

\smallskip

The values $\kappa_1^\gamma, \kappa_2^\gamma,\cdots$ will be called the \emph{quantized $\gamma$-cumulants} of the limiting measure $\mu^\gamma$. This terminology is used because, upon certain normalization followed by the limit $\gamma\to\infty$, the quantities $\kappa_n^\gamma$ turn into the coefficients of the quantized $R$-transform from~\cite{BufetovGorin2015b}, as explained in \cite[Remark 8]{CuencaDolega2025}.
We showed in~\cite{CuencaDolega2025} that they determine and, at the same time, are uniquely determined by the moments $m_1,m_2,\dots$ of $\mu^\gamma$, according to the following recursive relations
\begin{multline}\label{eq:intro_moms_cums}
m_\ell =  \sum_{\Gamma\in\Luk(\ell)}
\frac{\Delta_\gamma\left(x^{1+\text{\#\,horizontal steps at height $0$ in $\Gamma$}}\right)\big|_{x=\kappa_1^\gamma}}{1+\text{\#\,horizontal steps at height $0$ in $\Gamma$}}
\,\prod_{i\ge 1}\big(\kappa_1^\gamma + i\big)^{\text{\#\,horizontal steps at height $i$ in $\Gamma$}}\\
\cdot\prod_{j\ge 1}\big(\kappa_j^\gamma + \kappa_{j+1}^\gamma\big)^{\text{\#\,steps $(1,j)$ in $\Gamma$}}(j+\gamma)^{\text{\#\,down steps from height $j$ in $\Gamma$}},\quad\text{ for all }\ell\in\Z_{\ge 1},
\end{multline}
where $\Luk(\ell)$ is the set of Łukasiewicz paths of length $\ell$ (certain lattice paths in $\Z^2$ whose definition is recalled in \cref{sec:recall_LLN}), and $\Delta_\gamma(f)(x) := \frac{1}{\gamma}(f(x)-f(x-\gamma))$.
% These equations are of the form
% \[
% m_\ell = C_\ell\kappa_\ell^\gamma + P_\ell(\kappa_1^\gamma,\dots,\kappa_{\ell-1}^\gamma),\quad\text{for all }\ell\in\Z_{\ge 1},
% \]
% for certain constants $C_\ell>0$ and polynomials $P_\ell(x_1,\dots,x_{\ell-1})$, therefore, the sequence $\big(\kappa_n^\gamma\big)_{n\ge 1}$ is uniquely determined by the moment sequence $(m_n)_{n\ge 1}$.
% Our main theorem, stated next, is a pair of identities of formal power series that gives an equivalent relation between these two sequences.

The generating series
\begin{equation*}
  R^\gamma_\mu(z) := \sum_{n\ge 1}{ \frac{\kappa_n^\gamma z^n}{n} }
\end{equation*}
for the quantized $\gamma$-cumulants of $\mu$ will be called the \emph{$\gamma$-quantized $R$-transform of $\mu$}. Equation~\ref{eq:intro_moms_cums} provides a combinatorial recipe for calculating this transform for measures with finite moments of all orders.
Then a natural question is whether the definition of $\gamma$-quantized $R$-transform can be extended to larger classes of measures, possibly by treating it as a complex analytic function.
In a similar vein, but for free cumulants instead of quantized $\gamma$-cumulants, the theory of Free Probability defines the $R$-transform as an analytic function that satisfies a certain functional equation involving the Cauchy transform.
The main theorem of this paper is a similar functional equation between the $\gamma$-quantized $R$-transform and the moment generating function of any probability measure with all finite moments.

\begin{theorem}[\cref{thm:r_transform} in the text]\label{thm:intro_3}
Assume that the sequences $(m_n)_{n\ge 1}$ and $\big(\kappa^\gamma_n\big)_{n\ge 1}$ satisfy the relations~\eqref{eq:intro_moms_cums}.
Then there exists a unique sequence $\big(c^\gamma_n\big)_{n\ge 1}$ such that we have the following equalities of formal power series in $z$ and $z^{-1}$, respectively,
\begin{equation}\label{eq:intro_display_formalLap}
\begin{aligned}
\exp\left(\sum_{n\geq 1}\frac{\kappa_n^\gamma z^n}{n}\right) &= 
1+\sum_{n\ge 1}\frac{c_n^\gamma}{\gamma^{\uparrow n}}z^n,\\
1+\sum_{n\ge 1}\frac{(-1)^n c_n^\gamma}{z^{\uparrow n}} &=
\exp\left(\gamma\cdot\Lapformal\left\{\frac{t}{1-e^{-t}}\sum_{n\geq 1}\left(\frac{(-1)^n m_n}{n!}-\frac{\gamma^n}{(n+1)!}\right)t^{n-1}\right\}(z)\right),
\end{aligned}
\end{equation}
where $z^{\uparrow n} := \prod_{i=0}^{n-1}{(z+i)}$ and $\Lapformal\left\{ \sum_{n\ge 0}{\frac{a_n}{n!} t^n} \right\}(z) := \sum_{n\ge 0}{a_n z^{-n-1}}$.
\end{theorem}

Although the combinatorial transform \eqref{eq:intro_moms_cums} can sometimes be used to show the uniqueness of $\mu$, it is not clear how to use it to study other properties of $\mu$.
Some special cases of combinatorial transforms coming from Viennot's construction~\cite{Viennot1985} are amenable to analysis through the associated orthogonal polynomials, but~\eqref{eq:intro_moms_cums} does not fit into this framework, therefore we needed to find another approach. Our \cref{thm:intro_3} enables us to analyze the limiting measure using analytic tools, and we apply it to study the high temperature limits of the \emph{pure Jack measures} to be discussed later. Indeed, \cref{thm:intro_3} will provide a precise formula for the characteristic function of the limiting measure and will lead to a formula for its density, which displays a remarkable crystallization phenomenon. Additionally, we hope that \cref{thm:intro_3} will become a key tool to extend the $\gamma$-quantized $R$-transform to measures with heavy tails in future research.

\subsection{Crystallization phenomenon for non-intersecting random walks at high temperature}

An important class of examples that can be understood through our $\gamma$-quantized $R$-transform are certain ensembles of discrete $N$-particle random point configurations with log-gas type interaction that depend on the inverse temperature parameter $\beta>0$.
One such motivating example derives from the continuous-time, discrete-space, integrable Markov process introduced and studied by Gorin--Shkolnikov~\cite{GorinShkolnikov2015}, which has a diffusive limit to the $\beta$-Dyson Brownian motion.
This process can be regarded as a one-parameter deformation, non-intersecting $N$-particle version of the Poisson random walk on the integer lattice; see \cref{fig:Markov_chain} for three simulations of this process, with different values of $\theta$.

\begin{figure}[htbp]
  \centering
  \begin{subfigure}[b]{0.325\textwidth}
    \centering
    \includegraphics[width=\textwidth]{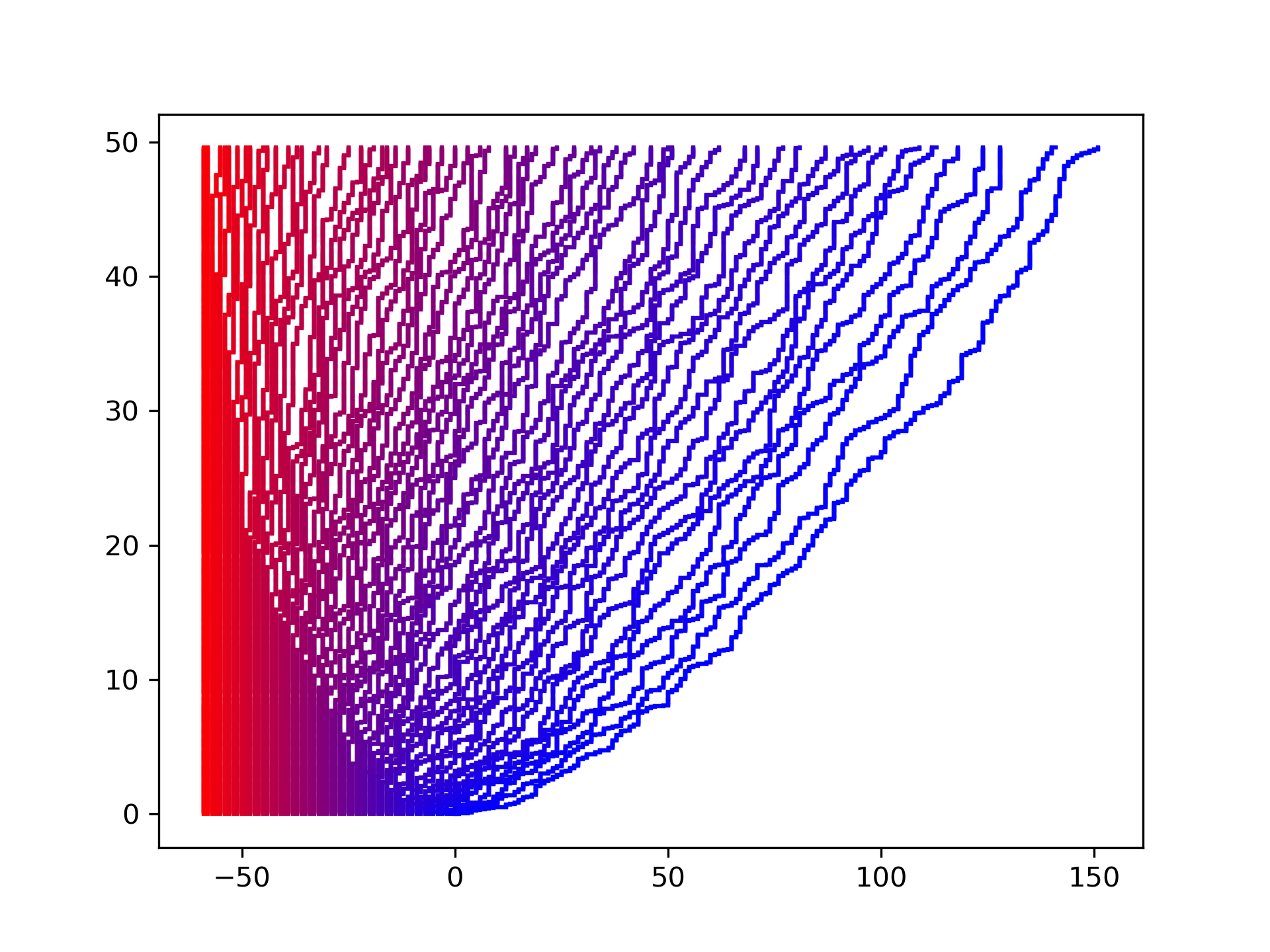}
    \caption{$N=60,\,\theta=1$}
    \label{fig:left}
  \end{subfigure}
  \hfill
  \begin{subfigure}[b]{0.325\textwidth}
    \centering
    \includegraphics[width=\textwidth]{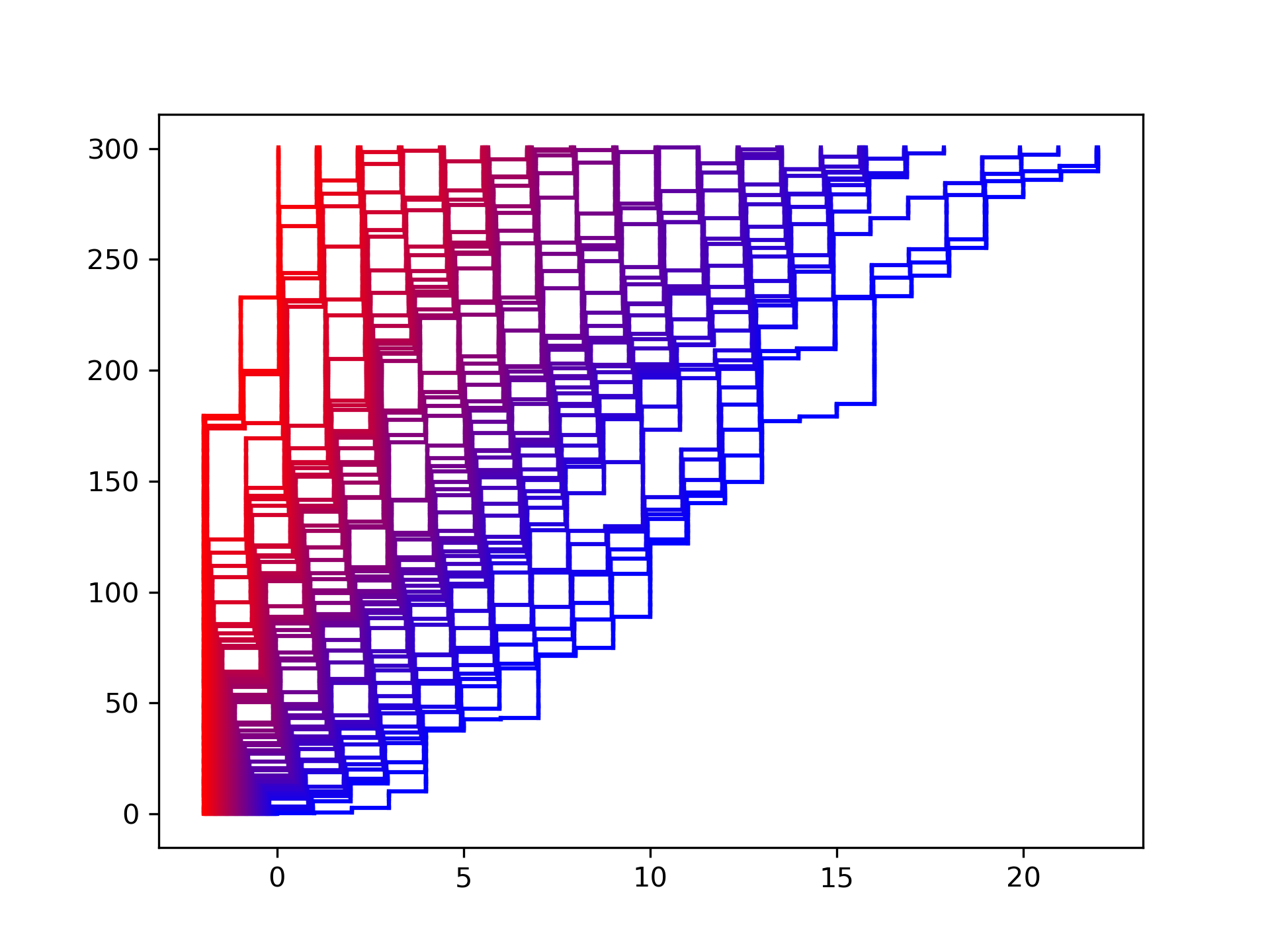}
    \caption{$N=60,\,\theta=\frac{2}{N}$}
    \label{fig:center}
  \end{subfigure}
  \hfill
  \begin{subfigure}[b]{0.325\textwidth}
    \centering
    \includegraphics[width=\textwidth]{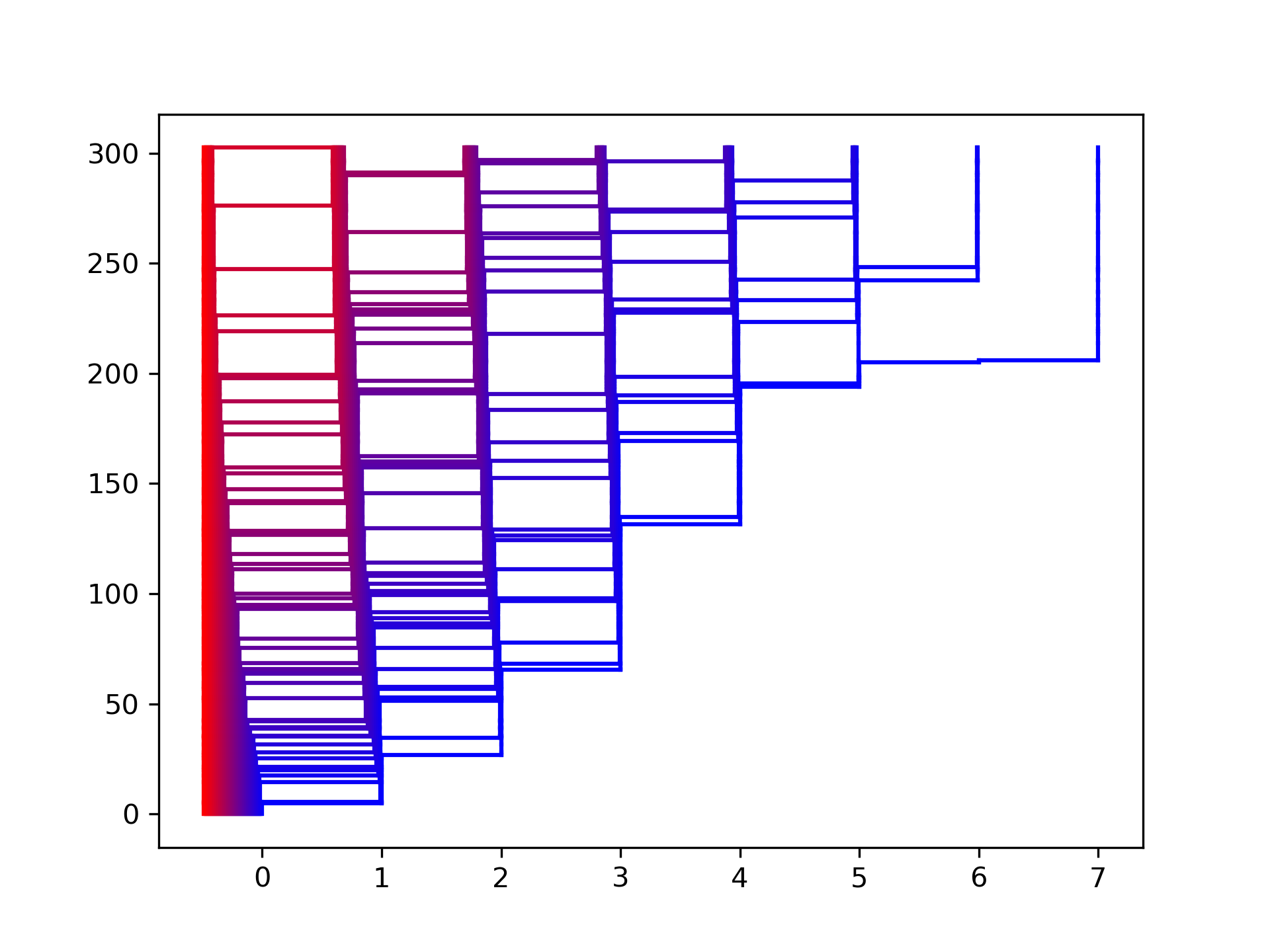}
    \caption{$N=60,\,\theta=\frac{1}{2N}$}
    \label{fig:right}
  \end{subfigure}
  \caption{Three simulations of the Gorin--Shkolnikov process with $N=60$ particles and trivial initial configuration. The simulations show the asymptotic behavior in the fixed temperature regime on the left with $\theta=1$, and in the high temperature regime with $\gamma = 2$ in the middle, and $\gamma = \frac{1}{2}$ on the right. The $x$-axis represents the space and the $y$-axis represents the time.}
  \label{fig:Markov_chain}
\end{figure}

The construction of~\cite{GorinShkolnikov2015} was based on the Plancherel specialization of the ring of symmetric functions.
Later, Huang~\cite{Huang2021} defined two other variations that result from replacing the Plancherel by the pure alpha and pure beta specializations.
They are non-intersecting $N$-particle versions of the geometric and Bernoulli random walks on $\Z$, respectively.
The LLN was proved in \cite{Huang2021,CuencaDolega2025} for \emph{fixed temperature $\theta^{-1}$}, general initial conditions, and appropriately chosen sequences of times; this can be appreciated in the left panel of \cref{fig:Markov_chain}.

\smallskip

In \cite{CuencaDolega2025}, we investigated the \emph{high temperature limits} of the fixed-time distributions of the previous three families of Markov processes. We proved that for general initial conditions, the answer depends on two ingredients: one is a new operation of probability measures, called the \emph{quantized $\gamma$-convolution}, explained in \cite{CuencaDolega2025}, and secondly, on the limiting measure when the Markov process is started at the trivial initial condition.
Moreover, we also proved that the \emph{pure Jack measures} (to be discussed in the next subsection) are the fixed-time distributions of the Markov processes started at the trivial initial conditions, therefore it remains to find the limiting measures of pure Jack measures. We proved that these limiting measures are uniquely determined by their moments, but solving the moment problem was beyond our reach.
Looking at the simulations in the middle and right panels of \cref{fig:Markov_chain}, however, we observed that these limiting measures have the following \emph{crystallization phenomenon}:~the large fixed-time distributions of the Markov processes appear to converge to a mixture of uniform distributions supported on the union of countably many intervals. By running a simulation with a larger number of particles and looking at the position of particles at specific times, this phenomenon becomes more visible; see \cref{fig:Markov_Simulation} below.
It seems that the support of the limiting measure has gaps of length $1$, so it is determined by the right endpoints of these intervals. In this paper, we prove that this is indeed true, and we fully describe the crystallization phenomenon for all three cases of pure Jack measures.

\begin{figure}[htbp]
  \centering
    \includegraphics[width=\textwidth]{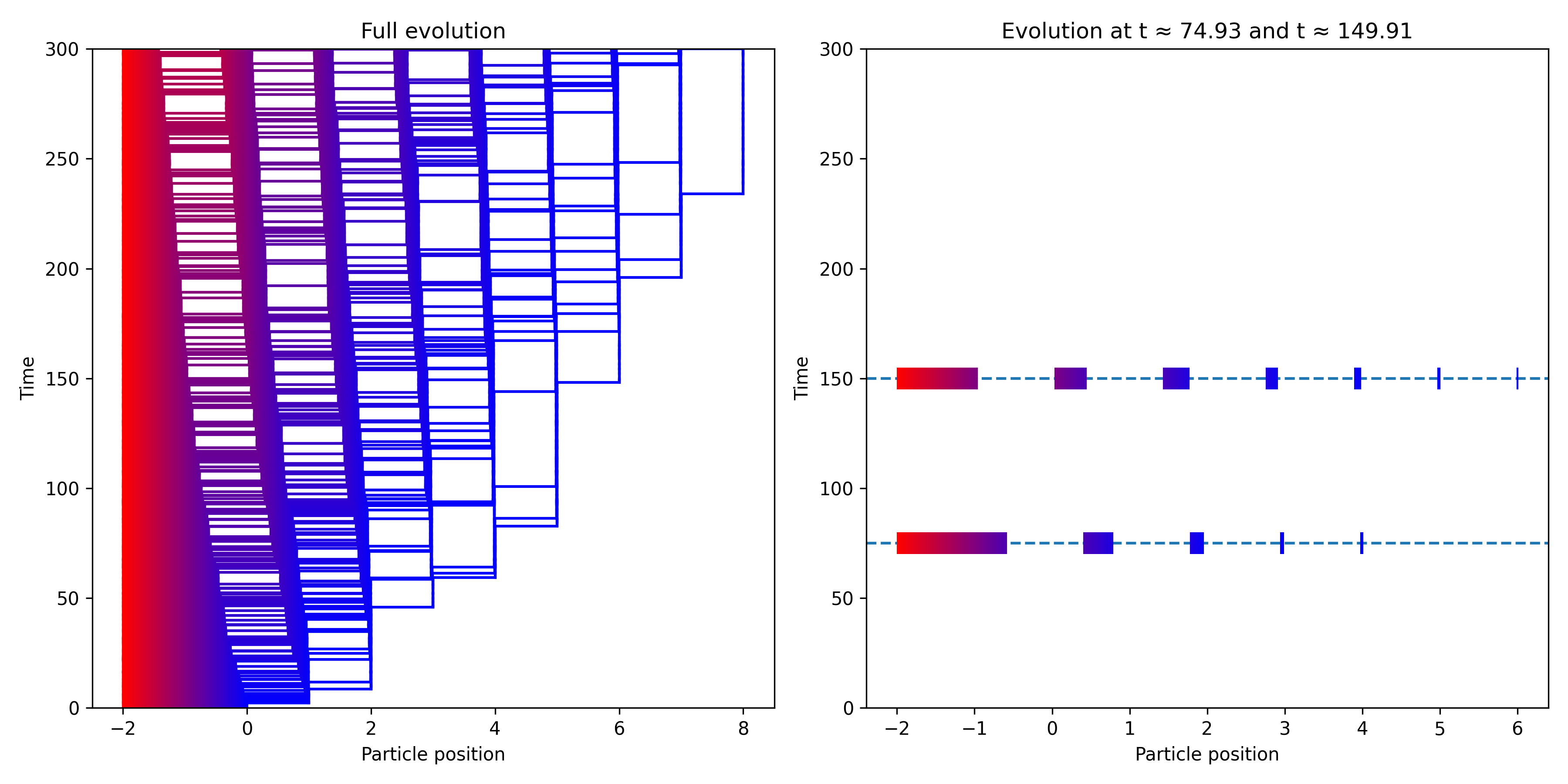}
  \caption{A simulation of the Gorin--Shkolnikov process with $N=300$ particles and trivial initial configuration $\LL^{(0)}_i = \theta(1-i)$, for $1 \leq i \leq N$, in the high temperature regime with $\theta = \frac{2}{N}$. Note that the initial configuration converges to the uniform distribution on $[-2,0]$. On the right-hand side, we see the distributions of particles at times $t\approx\frac{\eta}{\theta}$, for $\eta = 1/2$ and $\eta = 1$, which seem to be supported on the union of disjoint intervals with gaps between consecutive ones having length close to $1$.}
  \label{fig:Markov_Simulation}
\end{figure}

\begin{figure}[htbp]
  \centering
    \includegraphics[width=0.8\textwidth]{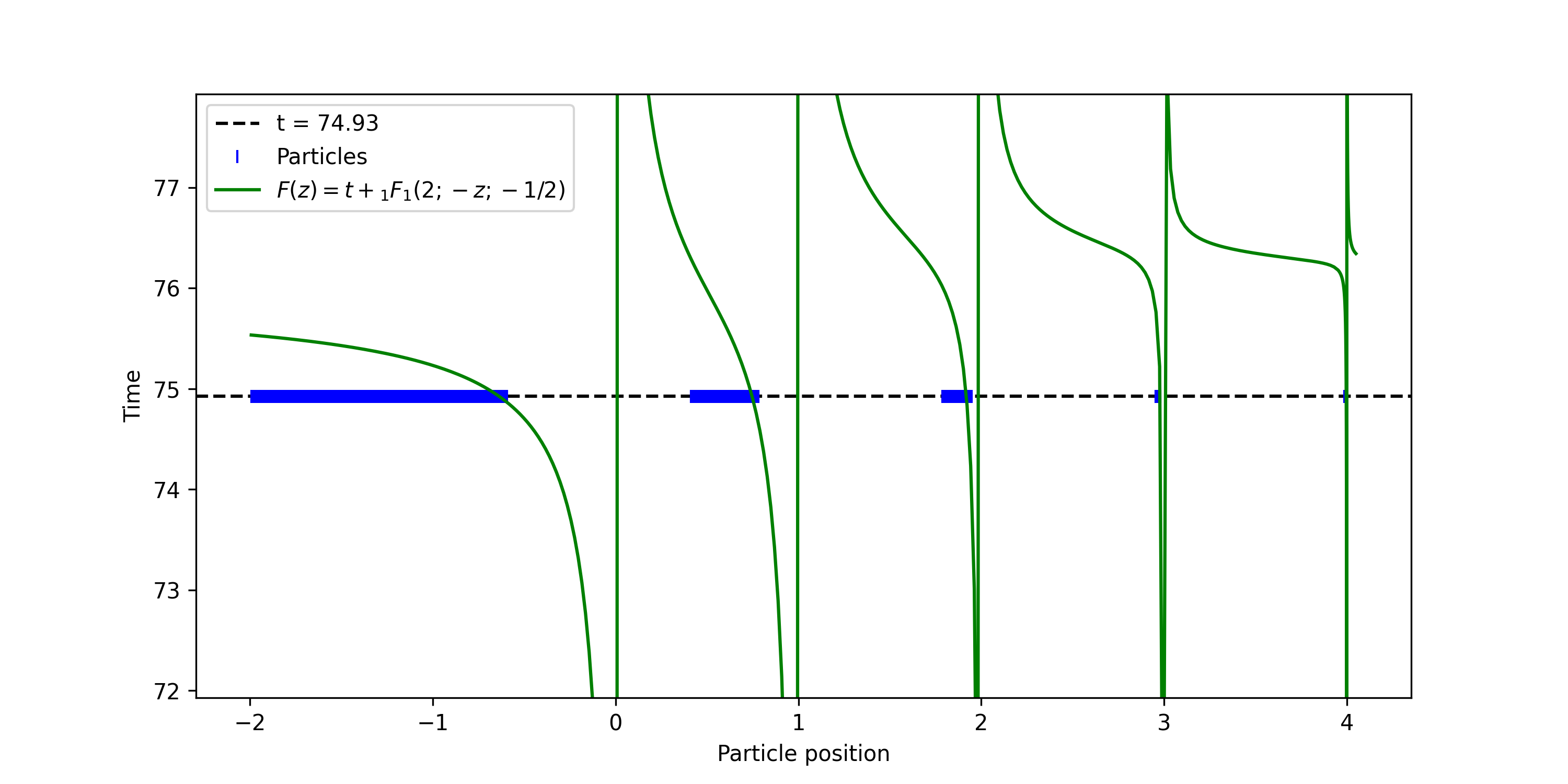}
  \caption{Juxtaposition of the simulation (in blue) from \cref{fig:Markov_Simulation} at time $t\approx\frac{\eta}{\theta}$ and the graph (in green) of the confluent hypergeometric function ${}_1F_1(\gamma;z;-\eta)$, where $\eta = 1/2$, in the real variable $z$. \cref{thm:application_plancherel} predicts that the right endpoints of the intervals on the picture approximate the zeroes of ${}_1F_1(\gamma;z;-\eta)$, which are all real.}
  \label{fig:Comparison1/2}
\end{figure}

\begin{figure}[htbp]
  \centering
    \includegraphics[width=0.8\textwidth]{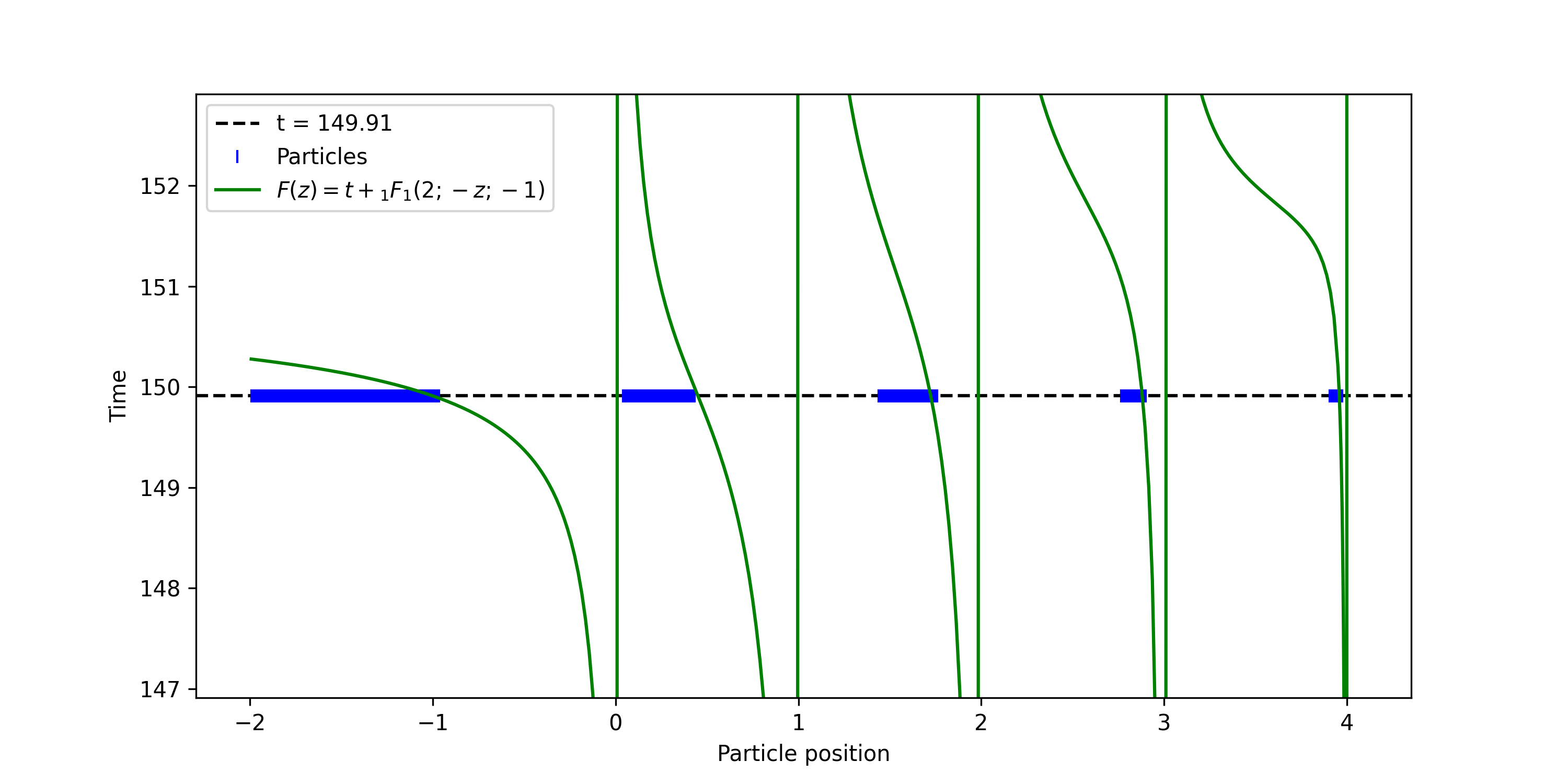}
  \caption{The same comparison as in \cref{fig:Comparison1/2} for $\eta = 1$.}
  \label{fig:Comparison1}
\end{figure}

\subsubsection{Pure Jack measures}

Denote the Jack symmetric functions by $P_\la(;\theta)$, and their $Q$-normalizations by $Q_\la(;\theta)$.
Given a Jack-positive specialization of the ring of symmetric functions $\rho\colon\Sym\to\R$ and $N\in\Z_{\ge 1}$, the formula
\[
\PP_{\rho;N}^{(\theta)}(\la) \propto Q_\la(\rho;\theta)P_\la(1^N;\theta)
\]
defines a probability measure on the set $\Y(N)$ of partitions of length at most $N$.
Given the classification result for Jack-positive specializations, there are three natural choices of $\rho$, namely, the pure alpha, pure beta, and Plancherel specializations.\footnote{The necessary definitions are recalled in \cref{sec:jacks}.}
The resulting probability measures on partitions will be called the \emph{pure Jack measures}.

\smallskip

We are interested in the empirical measures of sequences $\big\{\PP_{\rho_N;\,N}^{(\theta_N)}\big\}_{N\ge 1}$ in the high temperature regime
\[
\theta_N\to 0,\quad N\theta_N\to\gamma,\quad\text{as }\, N\to\infty,\ \text{for some }\gamma\in(0,\infty),
\]
where all $\rho_N$ remain within the pure alpha, pure beta, or Plancherel families of Jack-positive specializations. In our previous paper \cite{CuencaDolega2025}, we proved that the corresponding limiting measures $\mu^{\gamma;c;\eta}_{\Alpha}, \mu^{\gamma;c;M}_{\Beta}$ and $\mu^{\gamma;\eta}_{\Planch}$, are uniquely determined by their moments, and we gave combinatorial formulas for those moments. In this paper, we go a step further and employ our main \cref{thm:intro_3} to calculate the densities of the limiting measures, thereby proving the aforementioned crystallization phenomenon. The pure Plancherel Jack measure should be regarded as the discrete version of the Gaussian $\beta$-ensemble.
The papers \cite{DuyShirai2015,AllezBouchaudGuionnet2012} obtain a density for the high temperature limit of the G$\beta$E; see also \cite[Sec.~4.4]{Benaych-GeorgesCuencaGorin2022}.
Likewise, the pure alpha Jack measures are discrete versions of the Laguerre $\beta$-ensembles, and for the latter ones, the density was computed in \cite{AllezBouchaudMajumdarVivo2013}; see also \cite[Sec.~4.1]{Benaych-GeorgesCuencaGorin2022}.
Our \cref{thm:application_plancherel,thm:application_alpha} are discrete analogues, though the densities are qualitatively different.
Finally, we note that the pure beta Jack measure does not have a continuous analogue, therefore \cref{thm:application_beta} also does not.
This is because the classification result in \cref{thm:KOO} contains the pure beta specializations, while the continuous classification theorem \cite{OlshanskiVershik1996,AssiotisNajnudel2021} that is related to the Multivariate Bessel Functions of type A does not have these parameters. It is worth mentioning that the combinatorial formulas for the moments of the limiting measures in the continuous and discrete settings are similar, but the limiting measures are very different, and in the discrete case they always admit the observed crystalized form as given in \cref{thm:application_plancherel,thm:application_alpha,thm:application_beta}.

\smallskip

Finally, let us discuss one case in detail, namely, the limiting empirical measure $\mu^{\gamma;\eta}_{\Planch}$ of pure Plancherel measures, which also describes the limiting distribution in the Gorin--Shkolnikov process at time $\approx\frac{\eta}{\theta_N}$ in the regime $N\theta_N \to \gamma, N\to \infty$.
For $\gamma,\eta>0$, let $J^{(\text{planch})}_{\gamma,\eta}$ be the Jacobi operator on $\ell^2(\Z_{\ge 1})$ whose matrix has diagonal entries $\la_n = 1-\eta-n$, and off-diagonal entries $w_n = -\sqrt{\eta(\gamma+n)}$, for all $n\in\Z_{\ge 1}$.
Then we prove in \cref{sec:Jacobi,sec:analytic,sec:densities} that:

\begin{enumerate}
\item[\it{(i)}] the spectrum of $J^{(\text{planch})}_{\gamma,\eta}$ is discrete and consists of simple eigenvalues that can be labeled $\ell_1^{(\gamma;\eta)}>\ell_2^{(\gamma;\eta)}>\cdots$;

\item[\it{(ii)}] the zeroes of the entire function $F^{\text{planch}}_{\gamma,\eta}(z)=\frac{1}{\Gamma(z)}\cdot{}_1F_1(\gamma;z;-\eta)$ are all real, coincide with the spectrum of $J^{(\text{planch})}_{\gamma,\eta}$, and satisfy the inequalities
\[
-\gamma\le -\ell_1^{(\gamma;\eta)},\quad 1-\ell_k^{(\gamma;\eta)}\le -\ell_{k+1}^{(\gamma;\eta)}, \text{ for all }k\ge 1;
\]

\item[\it{(iii)}] a density for the limiting measure of (empirical measures of) the pure Plancherel Jack measures is
\[
\frac{\dd\mu^{\gamma;\eta}_{\Planch}(x)}{\dd x} = \frac{1}{\gamma}\left\{ \mathbf{1}_{\big[-\gamma,\,-\ell_1^{(\gamma;\eta)}\big]}(x) + \sum_{k=1}^\infty{ \mathbf{1}_{\big[1-\ell_k^{(\gamma;\eta)},\, -\ell_{k+1}^{(\gamma;\eta)}\big]}(x) } \right\}.
\]
\end{enumerate}

\subsection*{Acknowledgements}

CC is grateful to the Institute of Mathematics of the Polish Academy of Sciences (IMPAN) for hosting him during his research visits. MD would like to thank to Sławomir Dinew for stimulating discussions.

\section{Jack polynomials and Jack generating functions}\label{sec:jacks}

In this section, we briefly recall some useful material from \cite[Chapters~I \& VI.10]{Macdonald1995}, \cite{Stanley1989} and \cite{KerovOkounkovOlshanski1998}.
We largely follow the presentation of \cite[Sec.~2]{CuencaDolega2025}, and refer the reader to this article for further details.

\subsection{Partitions}

A partition is a sequence $\la=(\la_1\ge\la_2\ge\cdots)\in(\Z_{\ge 0})^\infty$ such that $\la_i=0$, for all large $i$. The size of $\la$ is $|\la|:=\sum_{i\ge 1}{\la_i}$, while its length is $\ell(\la):=\max\{ i\ge 1 \colon \la_i\ne 0 \}$.
The length of the empty partition $\emptyset=(0,0,0,\cdots)$ is zero, by convention. For a partition $\la$ with $\ell(\la)\le N$, we will write $\la=(\la_1,\dots,\la_N)$, with the understanding that infinitely many zeroes are being omitted.
For any $N\in\Z_{\ge 0}$, denote
\begin{equation*}
\Y(N) := \{ \la\in\Y \colon \ell(\la)\le N \},
\end{equation*}
and $\MM(\Y(N))$ the simplex of probability measures on $\Y(N)$.
We also let $\Y=\bigcup_{N\ge 0}{\Y(N)}$ be the set of all partitions.
We will often identify partitions with their Young diagrams.

\subsection{Algebra of symmetric functions}

We let $\Sym$ be the real algebra of symmetric functions in the variables $\xx=(x_1,x_2,\dots)$, i.e., this is the polynomial algebra $\Sym=\R[p_1,p_2,\dots]$ in the power sum variables $p_k=\sum_{i\ge 1}{x_i^k}$.

Upon the specialization $x_{N+1}=x_{N+2}=\cdots=0$, any $f\in\Sym$ becomes a symmetric polynomial $f(x_1,\dots,x_N)\in\R[x_1,\dots,x_N]^{\Sy{N}}$ and, conversely, for any sequence $\big\{f_N(x_1,\dots,x_N)\in\R[x_1,\dots,x_N]^{\Sy{N}}\big\}_{N\ge 1}$ of symmetric polynomials such that
\begin{equation}\label{eq:stability_property}
f_{N+1}(x_1,\dots,x_N,0)=f_N(x_1,\dots,x_N),\quad\text{for all }N\ge 1,
\end{equation}
their projective limit leads to a symmetric function $f(\xx)=\varprojlim{f_N(x_1,\dots,x_N)}\in\Sym$.

\subsection{Jack symmetric functions}

For any $N\in\Z_{\ge 1}$, the \emph{Jack polynomials} $P_\la(x_1,\dots,x_N;\theta)$, for $\la\in\Y(N)$, are defined uniquely by certain properties of triangularity (w.r.t. the monomial symmetric functions) and orthogonality (w.r.t. an explicit weight function); see~\cite[Ch.~VI.10]{Macdonald1995} for details.
They satisfy the stability property~\eqref{eq:stability_property}, which shows the existence of the projective limits $P_\la(\xx;\theta)=\varprojlim{P_\la(x_1,\dots,x_N;\theta)}$, called the \emph{Jack symmetric functions}.
Sometimes we denote the Jack symmetric functions by $P_\la(;\theta)$, without specifying the alphabet $\xx=(x_1,x_2,\dots)$.
It is known that $\{P_\la(;\theta)\}_{\la\in\Y}$ is a basis of $\Sym$.

The Jack symmetric functions satisfy the following \emph{Cauchy identity}
\begin{equation}\label{eqn:cauchy}
\sum_{\la\in\Y}P_\la(\xx;\theta)Q_\la(\yy;\theta) = \exp\Bigg\{ \theta\sum_{k\ge 1}{\frac{p_k(\xx)p_k(\yy)}{k}} \Bigg\} =: H_\theta(\xx;\yy),
\end{equation}
where $\xx=(x_1,x_2,\dots)$, $\yy=(y_1,y_2,\dots)$, and
\begin{equation}\label{j_lambda}
Q_\la(;\theta) := \prod_{(i,j)\in\la}\frac{(\la_i-j)+\theta(\la_j'-i)+\theta}{(\la_i-j)+\theta(\la_j'-i)+1}\cdot P_\la(;\theta)
\end{equation}
is the $Q$-normalization of Jack symmetric functions.
In~\cref{j_lambda}, we denoted $\la'$ the conjugate partition to $\la$, so that $\la_j'$ is the length of the $j$-th largest column of the Young diagram of $\la$.

%For any $\la,\mu\in\Y$, we also define the \emph{skew-Jack symmetric function} $Q_{\la/\mu}(\xx;\theta)$ (and analogously, the $P$-normalized version $P_{\la/\mu}(\xx;\theta)$) by the following identity:
%\begin{equation}\label{Q_skew}
%Q_\la(\xx,\yy;\theta) := \sum_{\mu \in \Y}Q_\mu(\yy;\theta)Q_{\la/\mu}(\xx;\theta).
%\end{equation}
%The sum on the right-hand side of~\eqref{Q_skew} is actually finite, since it is known that $Q_{\la/\mu}(\xx;\theta)=0$ unless $\mu\subseteq\la$.
%Moreover, $Q_{\la/\mu}(\xx;\theta)$ coincides with $Q_{\la}(\xx;\theta)$ when $\mu=\emptyset$.
%We will also need the following \emph{skew-Cauchy identity}:
%\begin{equation}\label{eqn:cauchy_skew}
%\sum_{\la\in\Y}P_{\la/\mu}(\xx;\theta)Q_{\la/\rho}(\yy;\theta) = H_\theta(\xx;\yy)\sum_{\la\in\Y}Q_{\mu/\la}(\yy;\theta)P_{\rho/\la}(\xx;\theta).
%\end{equation}

\subsection{Jack generating functions}

For any $N\in\Z_{\ge 1}$ and $\PP_N\in\M(\Y(N))$, define the \emph{Jack generating function} of $\PP_N$ as the following formal power series in $N$ variables:
\begin{equation}\label{eq:JackGeneratFunction}
G_{\PP_N,\theta}(x_1,\dots,x_N) := \sum_{\la\in\Y(N)}\PP_N(\la)\frac{P_\la(x_1,\dots,x_N;\theta)}{P_\la(1^N;\theta)}.
\end{equation}

\begin{definition}\label{def:small_tails}
We say that $\PP_N\in\M(\Y(N))$ has \textbf{small tails} if
\begin{equation}\label{eqn:laurent_series}
\sum_{\la\in\Y(N)}{ \PP_N(\la)\frac{P_\la(z_1,\dots,z_N;\theta)}{P_\la(1^N;\theta)} }
\end{equation}
converges absolutely on an $N$-dimensional disk of the form
\begin{equation}\label{eqn:annulus}
D_{N;R} := \{(z_1,\dots,z_N)\in\C^N\colon |z_i|<R,\text{ for all }i=1,\dots,N\},
\end{equation}
for some $R>1$.
\end{definition}

For example, if $\PP_N$ is supported on a finite subset of $\Y(N)$, then $\PP_N$ has small tails.
Note that if $\PP_N$ has small tails, then the formal power series~\eqref{eqn:laurent_series} defining $G_{\PP_N,\theta}(x_1,\dots,x_N)$ defines an analytic function in~\eqref{eqn:annulus}, and moreover $G_{\PP_N,\theta}(1^N) = 1$.

\subsection{Jack-positive specializations}\label{sec:jack_specs}

\begin{definition}
A \textbf{specialization} of $\Sym$ is a unital algebra homomorphism $\rho\colon\Sym\to\R$.
For any $f\in\Sym$, denote the image of $f$ under $\rho$ by $f(\rho)$.
The specialization $\rho\colon\Sym\to\R$ is said to be \textbf{Jack-positive} if
\begin{equation*}
P_\la(\rho;\theta) \ge 0,\quad\text{for all }\la\in\Y.
\end{equation*}
Note that this definition depends on the value of $\theta$, which is usually understood from the context.
\end{definition}

\begin{theoremdefinition}[\cite{KerovOkounkovOlshanski1998}]\label{thm:KOO}
Given $\theta>0$ fixed, the set of Jack-positive specializations is in bijection with points of the following Thoma cone:
\begin{multline*}
\Omega := \Big\{ (\alpha, \beta, \delta)\in(\R_{\ge 0})^\infty\times(\R_{\ge 0})^\infty\times\R_{\ge 0} \quad \Big| \quad \alpha=(\alpha_1,\alpha_2,\dots),\ \beta=(\beta_1,\beta_2,\dots),\\
\alpha_1\ge\alpha_2\ge\dots\ge 0,\quad \beta_1\ge\beta_2\ge\dots\ge 0,\quad \sum_{i=1}^\infty(\alpha_i+\beta_i) \le \delta \Big\}.
\end{multline*}
For any $\omega=(\alpha,\beta,\delta)\in\Omega$, the corresponding Jack-positive specialization $\rho_\omega\colon\Sym\to\R$ is
\begin{equation*}
p_1(\rho_\omega)=\delta,\qquad p_k(\rho_\omega)=\sum_{i=1}^\infty\Big( \alpha_i^k + (-\theta)^{k-1}\beta_i^k \Big),\quad k\ge 2.
\end{equation*}
If $\omega=(\alpha,\beta,\delta)\in\Omega$, we say that $\alpha_i, \beta_i,\delta$ are the \textbf{Thoma parameters} of $\rho_\omega$.
\end{theoremdefinition}

\begin{example}
\label{exam:pure_alpha}
If all Thoma parameters of $\omega=(\alpha,\beta,\delta)\in\Omega$ vanish, except possibly $\alpha_1,\dots,\alpha_n$ and $\delta = \sum_{i=1}^n\alpha_i$, then we say that $\rho_\omega$ is a \emph{pure alpha specialization} and we denote it by $\Alpha(\alpha_1,\dots,\alpha_n)$.
For any $f=f(x_1,x_2,\dots)\in\Sym$, we have
\begin{equation*}
f\big(\Alpha(\alpha_1,\dots,\alpha_n)\big) = f(\alpha_1,\dots,\alpha_n),
\end{equation*}
i.e. $\Alpha(\alpha_1,\dots,\alpha_n)$ specializes $n$ variables $x_i\mapsto\alpha_i$ and the rest of $x_i$'s are set to zero.
\end{example}

\begin{example}\label{exam:pure_beta}
If the only possibly nonzero Thoma parameters of $\omega=(\alpha,\beta,\delta)\in\Omega$ are $\beta_1,\dots,\beta_n$ and $\delta = \sum_{i=1}^n\beta_i$, then we say that $\rho_\omega$ is a \emph{pure beta specialization} and we denote it by $\Beta(\beta_1,\dots,\beta_n)$.
For any $f\in\Sym$, we have
\begin{equation*}
f\big(\Beta(\beta_1,\dots,\beta_n)\big) = (\omega_{\theta^{-1}}f)(\theta\beta_1,\dots,\theta\beta_n),
\end{equation*}
where $\omega_{\theta^{-1}}$ is the automorphism of $\Sym$ defined by $\omega_{\theta^{-1}}(p_r):=(-1)^{r-1}\theta^{-1}p_r$, for all $r\ge 1$.
This means that $\Beta(\beta_1,\dots,\beta_n)$ is the composition of $\omega_{\theta^{-1}}$ and the pure alpha specialization $\Alpha(\theta\beta_1,\dots,\theta\beta_n)$.
\end{example}

\begin{example}\label{exam:plancherel}
If all Thoma parameters of $\omega=(\alpha,\beta,\delta)\in\Omega$ vanish, except possibly $\delta$, then we say that $\rho_\omega$ is a \emph{Plancherel specialization} and we denote it by $\Planch(\delta)$.
Alternatively, it is uniquely defined by $p_k(\Planch(\delta)) = \delta\cdot\mathbf{1}_{\{k=1\}}$.
\end{example}

%For any specializations $\rho_1,\rho_2$ of $\Sym$, the \emph{union specialization}, to be denoted $(\rho_1,\rho_2)$, is defined by the formula
%\begin{equation*}
%p_k(\rho_1,\rho_2) := p_k(\rho_1) + p_k(\rho_2),\quad\text{for all }k\ge 1.
%\end{equation*}
%If $\rho_1,\rho_2$ are Jack-positive, then so is the union $(\rho_1,\rho_2)$, by virtue of \cref{thm:KOO}.

\section{Law of Large Numbers at high temperature for discrete N-particle ensembles}

We recall the main result and applications obtained in~\cite{CuencaDolega2025}.

\subsection{The statement of the LLN}\label{sec:recall_LLN}

We consider sequences of probability measures\footnote{In \cite{CuencaDolega2025}, our main theorem was more general, as it considered finite signed measures, instead of probability measures, on the sets of $N$-signatures, instead of partitions of length $\le N$. For simplicity of exposition, we restrict our setting in this paper.} $\{\PP_N\in\MM(\Y(N))\}_{N \geq 1}$ and study the \emph{high temperature regime}
\begin{equation}\label{eq:HTRegime}
N\to\infty,\quad \theta\to 0,\quad N\theta\to\gamma,
\end{equation}
where $\gamma\in(0,\infty)$ is a fixed constant.
We assume that all $\PP_N$ have small tails, so that their Jack generating functions, denoted
\begin{equation}\label{eqn:jack_G}
G_{N,\theta}(x_1,\dots,x_N) := \sum_{\la\in\Y(N)}\PP_N(\la)\frac{P_\la(x_1,\dots,x_N;\theta)}{P_\la(1^N;\theta)},
\end{equation}
are analytic in disks of radii larger than $1$.
In particular, it follows that the logarithms
\begin{equation}\label{eqn:jack_F}
F_{N,\theta}(x_1,\dots,x_N) := \ln(G_{N,\theta}(x_1,\dots,x_N))
\end{equation}
are also analytic on disks of radii greater than $1$.
Moreover, $F_{N,\theta}(1^N)=0$.

\begin{definition}\label{def:appropriate}
Let $\{\PP_N\in\MM(\Y(N))\}_{N \geq 1}$ be a sequence of probability measures with small tails and let $F_{N,\theta}=F_{N,\theta}(x_1,\dots,x_N)$ be the logarithms of their Jack generating functions, defined by \cref{eqn:jack_G,eqn:jack_F}.
We say that $\{\PP_N\}_{N \geq 1}$ is \textbf{HT-appropriate} if 

\begin{enumerate}
	\item $\displaystyle\lim_{\substack{N\to\infty\\N\theta\to\gamma}} \frac{1}{(n-1)!}\cdot\frac{\partial^n}{\partial x_1^n}\, F_{N,\theta}\big|_{(x_1,\dots,x_N) = (1^N)} = \kappa_n$ exists and is finite, for all $n \in \Z_{\geq 1}$,

	\item $\displaystyle\lim_{\substack{N\to\infty\\N\theta\to\gamma}} \frac{\partial^r}{\partial x_{i_1} \cdots \partial x_{i_r}}\, F_{N,\theta}\big|_{(x_1,\dots,x_N) = (1^N)} = 0$, for all $r\ge 2$ and $i_1,\dots,i_r\in\Z_{\geq 1}$ with at least two distinct indices among $i_1,\dots,i_r$.
\end{enumerate}
\end{definition}

We will sometimes regard measures $\PP_N$ on $\Y(N)$ also as measures on real $N$-tuples $(\LL_1>\dots>\LL_N)$, such that $\LL_i+(i-1)\theta\in\Z$, for all $i=1,\dots,N$, by using the shifted coordinates
\[
\LL_i := \la_i - (i-1)\theta, \ \text{ for all }i=1,\dots,N.
\]

\begin{definition}\label{def:LLN}
Let $\{\PP_N\in\MM(\Y(N))\}_{N \geq 1}$ be a sequence of probability measures with small tails.
We say that $\{\PP_N\}_{N \geq 1}$ \textbf{satisfies the LLN} if  there exist $m_1,m_2,\dots\in\R$ such that, for all $s\in\Z_{\ge 1}$ and $k_1,\dots,k_s\in\Z_{\ge 1}$, we have
\begin{equation}\label{eq:DefConvMoments}
\lim_{\substack{N\to\infty,\\ N\theta \to \gamma}}\frac{1}{N^s}\,\E_{\PP_N} \left[ \prod_{j=1}^s{ \sum_{i=1}^N{\LL_i^{k_j}} } \right]
= \prod_{j=1}^s{m_{k_j}},
\end{equation}
where, in the LHS, the $N$-tuples $(\LL_1>\dots>\LL_N)$ are $\PP_N$-distributed.
\end{definition}

The LHS of~\eqref{eq:DefConvMoments} can be interpreted in terms of the \emph{empirical measures} $\mu_N$ of $\PP_N$, which by definition are the (random) probability measures
\begin{equation}\label{eqn:empirical_measures}
\mu_N := \frac{1}{N}\sum_{i=1}^N\delta_{\LL_i},\text{ where }(\LL_1>\dots>\LL_N)\text{ is }\PP_N\text{--distributed}.
\end{equation}
Indeed, the moments of the empirical measures are
\begin{equation*}
m_k(\mu_N) := \int_{\R}x^k\mu_N(\text{d}x) = \frac{1}{N}\sum_{i=1}^N{\LL_i^k},
\end{equation*}
therefore~\eqref{eq:DefConvMoments} is equivalent to:
\begin{equation*}
\lim_{\substack{N\to\infty\\N\theta\to\gamma}} \E_{\PP_N} \left[ \prod_{j=1}^s{ m_{k_j}(\mu_N) } \right]
= \prod_{j=1}^s{m_{k_j}}.
\end{equation*}
As a result, LLN-satisfaction is equivalent to the convergence $\mu_N\to\mu$ in the sense of moments, in probability, to a probability measure $\mu$ with moments $m_1,m_2,\cdots$.
If $\mu$ is uniquely determined by its moments, it would also follow that $\mu_N\to\mu$ weakly, in probability.

\smallskip

A \emph{Łukasiewicz path $\Gamma$ of length $\ell$} is a lattice path on $\Z^2$ that starts at $(0,0)$, ends at $(\ell,0)$, stays above the $x$-axis, with steps $(1,0)$ (horizontal steps), $(1,-1)$ (down steps), and steps $(1, j)$, for some $j\in\Z_{\ge 1}$ (up steps).
We denote by $\Luk(\ell)$ the set of all Łukasiewicz paths of length $\ell$.
For example, there are exactly five Łukasiewicz paths of length $3$; see \cref{fig:PathsEx}.

\begin{definition}\label{def:mk}
Given $\gamma\in (0,\infty)$ and $\vec{\kappa}=(\kappa_n)_{n\ge 1}$, define $\vec{m}=(m_n)_{n\ge 1}$ by
\begin{multline}\label{eq:Moments}
m_\ell := \sum_{\Gamma\in\Luk(\ell)}
\frac{\Delta_\gamma\left(x^{1+\text{\#\,horizontal steps at height $0$ in $\Gamma$}}\right)(\kappa_1)}{1+\text{\#\,horizontal steps at height $0$ in $\Gamma$}}
\cdot\prod_{i\ge 1}(\kappa_1+i)^{\text{\#\,horizontal steps at height $i$ in $\Gamma$}}\\
\cdot\prod_{j\ge 1}(\kappa_j+\kappa_{j+1})^{\text{\#\,steps $(1,j)$ in $\Gamma$}}(j+\gamma)^{\text{\#\,down steps from height $j$ in $\Gamma$}},\quad\text{for all $\ell\in\Z_{\ge 1}$},
\end{multline}
where the term $\Delta_\gamma$ is the divided difference operator, defined by $\Delta_\gamma(f)(x) := \frac{1}{\gamma}(f(x)-f(x-\gamma))$, followed by evaluation at $x = \kappa_1$.
We denote the map $\vec{\kappa}\mapsto\vec{m}$ by $\mathcal{J}_\gamma^{\kappa\mapsto m}\colon\R^\infty\to\R^\infty$.
\end{definition}

In general, \cref{eq:Moments} has the form 
\[
m_\ell = (\gamma+1)^{\uparrow(\ell-1)}\cdot\kappa_\ell\,+ \text{ polynomial over } \Q[\gamma] \text{ in }\kappa_1,\cdots,\kappa_{\ell-1},
\]
where
\[
g^{\uparrow 0} := 1,\qquad g^{\uparrow n} := \prod_{i=1}^{n}{(g+i-1)},\quad n\in\Z_{\ge 1},
\]
is the \emph{raising factorial}. As a result, we deduce the equations
\begin{equation}\label{eq:kappa-top}
\kappa_\ell = \frac{1}{(\gamma+1)^{\uparrow(\ell-1)}}\cdot m_\ell\,+ \text{ polynomial over } \Q(\gamma) \text{ in }m_1,\cdots,m_{\ell-1},
\end{equation}
i.e., $\mathcal{J}_\gamma^{\kappa\mapsto m}$ is invertible, and the sequences $(\kappa_n)_{n\ge 1}$, $(m_n)_{n\ge 1}$ uniquely determine each other.

\begin{example}\label{ex:mk}
Let $\vec{m}=\mathcal{J}_\gamma^{\kappa\mapsto m}(\vec{\kappa})$.
From \cref{fig:PathsEx}, the third moment $m_3$ is equal to
\begin{equation*}
m_3 = (\gamma+1)(\gamma+2)\kappa_3 + 3(\gamma+1)\kappa_2\kappa_1+\kappa_1^3+3(\gamma+1)\kappa_2+\frac{3}{2}(\gamma+2)\kappa_1^2+\kappa_1-\frac{\gamma^3}{4}.
\end{equation*}
\end{example}

\begin{figure}
	\centering
	\includegraphics[width=\textwidth]{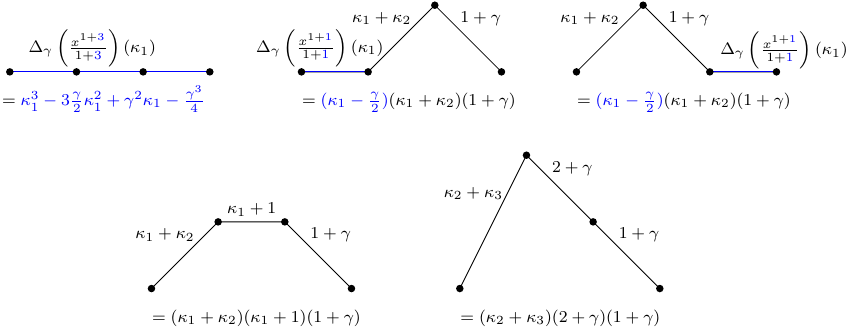}
	\caption{The five Łukasiewicz paths of length $3$, the weights associated to each of their steps, and the overall contribution after applying the divided difference operator $\Delta_\gamma$. The horizontal steps at height $0$ and their weight after applying the divided difference operator are marked in blue.}
	\label{fig:PathsEx}
\end{figure}

\begin{theorem}[Thm.~3.6 in \cite{CuencaDolega2025}]\label{theo:main1}
Let $\{\PP_N\in\MM(\Y(N))\}_{N \geq 1}$ be a sequence of probability measures with small tails.
Then $\{\PP_N\}_{N\ge 1}$ is HT-appropriate if and only if it satisfies the LLN.
If these equivalent conditions hold, then the sequences $\vec{\kappa}=(\kappa_n)_{n\ge 1}$ and $\vec{m}=(m_n)_{n\ge 1}$ are related to each other by $\vec{m}=\mathcal{J}_\gamma^{\kappa\mapsto m}(\vec{\kappa})$.
Furthermore, if there exists $C>0$ such that $|\kappa_n|\le C^n$, for all $n\geq 1$, then there exists a unique probability measure $\mu^\gamma$ such that
\begin{equation*}
m_\ell = \int_{\R}x^\ell \mu^\gamma(\text{d}x),\quad\text{for all }\ell\in\Z_{\ge 1}.
\end{equation*}
Finally, the empirical measures $\mu_N$ of $\PP_N$ converge to $\mu^\gamma$ weakly, in probability.
\end{theorem}

\subsection{Application to pure Jack measures}

\begin{definition}[\cite{BorodinOlshanski2005}]
For any Jack-positive specializations $\rho_1,\rho_2\colon\Sym\to\R$ satisfying
\begin{equation}\label{eq:finiteness}
\sum_{\la\in\Y}{Q_\la(\rho_1;\theta)P_\la(\rho_2;\theta)} < \infty,
\end{equation}
the \textbf{Jack measure} $\,\PP^{(\theta)}_{\rho_1;\rho_2}(\cdot)$ is the probability measure on the set of all partitions, defined by
\[
\PP^{(\theta)}_{\rho_1;\rho_2}(\la) := \frac{1}{H_\theta(\rho_1;\rho_2)}\cdot Q_\la(\rho_1;\theta)P_\la(\rho_2;\theta),\quad\la\in\Y,
\]
where $H_\theta(\rho_1;\rho_2)$ and $Q_\la(;\theta)$ are defined by \cref{eqn:cauchy,j_lambda}, respectively.
\end{definition}

We are interested in the Jack measures with $\rho_2=\Alpha(1^N)$ is the pure alpha specialization that sets $N$ variables equal to $1$ and the rest to zero (recall \cref{exam:pure_alpha}). In this case, the condition~\eqref{eq:finiteness}, for $\rho_1=:\rho$ and $\rho_2=\Alpha(1^N)$, is satisfied if the series
\begin{equation}\label{eq:stable_series}
\sum_{k\ge 1}\frac{|p_k(\rho)|}{k}z^k
\end{equation}
is absolutely convergent on some disk of radius larger than $1$.
Then the Jack measure $\PP^{(\theta)}_{\rho;\Alpha(1^N)}$ is well-defined and depends on one specialization $\rho$ and $N\in\Z_{\ge 1}$; moreover, it is supported on the subset $\Y(N)\subset\Y$, since $P_\la(1^N;\theta)\ne 0$, only if $\la\in\Y(N)$.
Let us denote this Jack measure by $\PP_{\rho;N}^{(\theta)}$, i.e.
\begin{equation}\label{eq:pure_jack_formula}
\PP_{\rho;N}^{(\theta)}(\la) = \exp\left( -\theta N\sum_{k\ge 1}{\frac{p_k(\rho)}{k}} \right) Q_\la(\rho;\theta)P_\la(1^N;\theta),\quad\la\in\Y(N),
\end{equation}
and denote its Jack generating function by $G_{\rho;N,\theta}(x_1,\dots,x_N)$.
As an application of the Cauchy identity \eqref{eqn:cauchy}, we deduce
\begin{equation}\label{eq:JGF_pure}
G_{\rho;N,\theta}(x_1,\dots,x_N) = \prod_{i=1}^N \exp\left\{ \theta\sum_{k\ge 1}\frac{p_k(\rho)}{k}\,(x_i^k-1)\right \}.
\end{equation}
If $\rho$ is a pure alpha, pure beta, or Plancherel specialization, we call $\PP_{\rho;N}^{(\theta)}$ a \emph{pure Jack measure}.
Formula~\eqref{eq:JGF_pure} allows us to prove the LLN, via \cref{theo:main1}, for the three families of pure Jack measures (see \cite[Sec.~5.2]{CuencaDolega2025}); in all cases, $\theta=\theta_N$ and $\rho=\rho_N$ depend on $N$.
Let us recall these results.

\subsubsection{Pure alpha Jack measures}

In this subsection, $c\in (0,1)$, $\eta>0$ are fixed.
The relevant pure alpha specialization is $\rho=\Alpha\big( c^{\lfloor\eta/\theta\rfloor} \big)$ (if $\theta$ depends on $N$, so does $\rho$), namely
\[
p_k\big( \Alpha\big( c^{\lfloor \eta/\theta \rfloor} \big) \big) = \lfloor \eta/\theta \rfloor\cdot c^k,\quad\text{for all }k\ge 1.
\]
Then the formula~\eqref{eq:pure_jack_formula} gives a probability measure on the set of partitions of length at most $\min\{N,\lfloor\eta/\theta\rfloor\}$, given explicitly by
\begin{multline*}
\PP^{(\theta)}_{\Alpha(c^{\lfloor\eta/\theta\rfloor});\,N}(\la) = (1-c)^{N\theta\lfloor\eta/\theta\rfloor} c^{|\la|}\\
\times\prod_{(i,j)\in\la}{ \frac{ (\lfloor \eta/\theta\rfloor\theta + (j-1) - \theta(i-1)) (N\theta + (j-1) - \theta(i-1)) }
{ ((\la_i-j)+\theta(\la_j'-i)+\theta)((\la_i-j)+\theta(\la_j'-i)+1) } },
\end{multline*}
for all $\la\in\Y\big(\min\{N, \lfloor\eta/\theta\rfloor\}\big)$. These measures will be called \emph{pure alpha Jack measures}.\footnote{The pure alpha Jack measures in \cite{CuencaDolega2025} slightly differ from those here, but can be recovered by the change of variables $\eta\mapsto\eta N\theta$.}

\begin{theorem}[Thm.~5.11 in \cite{CuencaDolega2025}]\label{thm:application1}
    Let $\{\theta_N\}_{N\ge 1}$ be a sequence of positive real numbers such that $N\theta_N\to\gamma$, as $N\to\infty$. Then the sequence of pure alpha Jack measures $\big\{\Prob^{(\theta_N)}_{\Alpha(c^{\lfloor\eta/\theta_N\rfloor});\,N}\big\}_{N\ge 1}$ satisfies the LLN and the corresponding empirical measures converge weakly, in probability, to the probability measure $\mu^{\gamma;c;\eta}_{\Alpha}$ that is uniquely determined by its moments, and that has quantized $\gamma$-cumulants
    \[
    \kappa_n^\gamma\Big[ \mu_{\Alpha}^{\gamma;c;\eta} \Big] = \frac{\eta c^n}{(1-c)^n},\quad\text{for all }n\in\Z_{\ge 1}.
    \]
    The moments of $\mu^{\gamma;c;\eta}_{\Alpha}$ are obtained from these quantized $\gamma$-cumulants and \cref{eq:intro_moms_cums}.
\end{theorem}

\subsubsection{Pure beta Jack measures}

In this subsection, $c\in (0,1)$, $M\in\Z_{\ge 1}$ are fixed.
The relevant specialization is $\rho=\Beta\big( (b/\theta)^M \big)$, where $b:=c/(1-c)$, namely
\[
p_k\big(\Beta\big((b/\theta)^M\big)\big) = (-1)^{k-1}\theta^{-1}Mb^k = \frac{(-1)^{k-1}Mc^k}{\theta(1-c)^k},\quad\text{for all }k\ge 1.
\]

\begin{theorem}[Thm.~5.14 in \cite{CuencaDolega2025}]\label{thm:application2}
Let $\{\theta_N>0\}_{N\ge 1}$ be such that $N\theta_N\to\gamma$, as $N\to\infty$.
Then the sequence of pure beta Jack measures $\big\{\PP^{(\theta_N)}_{\Beta((b/\theta_N)^M);\,N}\big\}_{N\ge 1}$ (where $b=c/(1-c)$) satisfies the LLN and their empirical measures converge weakly, in probability, to the probability measure $\mu_{\Beta}^{\gamma;c;M}$ uniquely determined by its moments and has quantized $\gamma$-cumulants
\[
\kappa_n^\gamma\Big[ \mu_{\Beta}^{\gamma;c;M} \Big] = (-1)^{n-1}Mc^n,\quad\text{for all }n\in\Z_{\ge 1}.
\]
\end{theorem}

\subsubsection{Pure Plancherel Jack measures}

In this subsection, $\eta>0$ is fixed. The relevant specialization is $\rho=\Planch(\eta/\theta)$, i.e.,
\[
p_k\big( \Planch(\eta/\theta) \big) = \frac{\eta}{\theta} \cdot\delta_{k,1},\quad\text{for all }k\ge 1.
\]
Then~\eqref{eq:pure_jack_formula} gives the probability measure
\begin{equation*}
\PP^{(\theta)}_{\Planch(\eta/\theta);\,N}(\la) = e^{-N\eta} \eta^{|\la|}
\prod_{(i,j)\in\la}{ \frac{N\theta + (j-1) - \theta(i-1)}{((\la_i-j) + \theta(\la_j'-i) + \theta)((\la_i-j) + \theta(\la_j'-i) + 1)} },
\end{equation*}
for all $\la\in\Y(N)$, called a \emph{pure Plancherel Jack measure}.

\begin{theorem}[Thm.~5.17 in \cite{CuencaDolega2025}]\label{thm:application3}
Let $\{\theta_N>0\}_{N\ge 1}$ be such that $N\theta_N\to\gamma$, as $N\to\infty$.
Then the sequence of pure Plancherel Jack measures $\big\{\PP^{(\theta_N)}_{\Planch(\eta/\theta_N);\,N}\big\}_{N\ge 1}$ satisfies the LLN and their empirical measures converge weakly, in probability, to the probability measure $\mu_{\Planch}^{\gamma;\eta}$ that is uniquely determined by its moments and has quantized $\gamma$-cumulants
\[
\kappa_n^\gamma\Big[ \mu_{\Planch}^{\gamma;\eta} \Big] = \eta\cdot\delta_{n,1},\quad\text{for all }n\in\Z_{\ge 1}.
\]
\end{theorem}

%\subsection{Non-intersecting random walks}

\section{The gamma-quantized R-transform and moment generating function}\label{sec:R_transform}

As a consequence of \cref{theo:main1}, we can define, for any $\gamma>0$, the sequence of quantized $\gamma$-cumulants $\{\kappa_n^\gamma[\mu]\}_{n\ge 1}$ for probability measures $\mu$ that are weak limits, in probability, of atomic measures of the form~\eqref{eqn:empirical_measures}.
In fact, the quantities $\kappa_n^\gamma[\mu]$ are obtained by first computing the moments $m_n(\mu)$ and then applying the recursive combinatorial formula from \cref{def:mk}.
There is certain parallelism to Speicher's combinatorial approach to Free Probability, where \emph{free cumulants} of probability measures are obtained from their moments by a recursive formula that involves noncrossing set partitions~\cite{Speicher1994}.

Speicher's approach to Free Probability, however, cannot be used to define free cumulants for probability measures that do not have finite moments of all orders.
This is where the \emph{analytic approach} to Free Probability shines, as instead of working with the sequence of cumulants, one works with the \emph{$R$-transform} and proves its existence as an analytic function for a large class of probability measures by means of a functional equation between the $R$-transform and Cauchy transform.
Envisioning a \emph{quantized $\gamma$-deformed Probability Theory}, we also want to find a functional equation between the series
\begin{equation*}
R^\gamma_\mu(z) := \sum_{n\ge 1}{ \frac{\kappa_n^\gamma[\mu]z^n}{n} },
\end{equation*}
which we call the \emph{$\gamma$-quantized $R$-transform of $\mu$}, and the Cauchy transform (or moment generating function) of $\mu$.
This motivated the following result.

\begin{theorem}\label{thm:r_transform}
Let $\gamma>0$ and assume that $\vec{\kappa}=(\kappa_1,\kappa_2,\dots)$ and $\vec{m}=(m_1,m_2,\dots)$ are related by the combinatorial identity from \cref{def:mk}, i.e.~$\vec{m}=\mathcal{J}_\gamma^{\kappa\mapsto m}(\vec{\kappa})$.
Then there exists a unique $\vec{c}=(c_1,c_2,\dots)$ such that the following relations between generating functions hold:
\begin{align}
\exp\left(\sum_{n\geq 1}\frac{\kappa_n z^n}{n}\right) &= 
1+\sum_{n\ge 1}\frac{c_n}{\gamma^{\uparrow n}}z^n,\label{eq:1}\\
1+\sum_{n\ge 1}\frac{(-1)^nc_n}{z^{\uparrow n}} &=
\exp\left(\gamma\cdot\Lapformal\left\{\frac{t}{1-e^{-t}}\sum_{n\geq 1}\left(\frac{(-1)^nm_n}{n!}-\frac{\gamma^n}{(n+1)!}\right)t^{n-1}\right\}(z)\right).\label{eq:2}
\end{align}
Recall that $z^{\uparrow n} = z(z+1)\cdots(z+n-1)$ is the raising factorial, while we define $\Lapformal$ as the linear operator that transforms power series in $t$ into power series in $z^{-1}$, and is uniquely defined by
\begin{equation}\label{eqn:formal_laplace}
\Lapformal\left\{ \sum_{n\ge 0}{\frac{a_n}{n!} t^n} \right\}(z) := \sum_{n\ge 0}{a_n z^{-n-1}}.
\end{equation}
\end{theorem}

Note that~\eqref{eq:1} is an equality between power series in $z$, which furnishes formulas between $(\kappa_1,\kappa_2,\dots)$ and $(c_1,c_2,\dots)$.
Likewise,~\eqref{eq:2} is an equality between power series in $z^{-1}$ that yields formulas between $(c_1,c_2,\dots)$ and $(m_1,m_2,\dots)$.
Put together, these two equalities can be used to express each $m_\ell$ in terms of $\kappa_1,\kappa_2,\dots$ (and viceversa); \cref{thm:r_transform} states that this gives an equivalent expression as the one in terms of Łukasiewicz paths from \cref{def:mk}.

\begin{remark}\label{rem:real_lap_transform}
We call $\Lapformal$ the \emph{formal Laplace transform} because it coincides with the Laplace transform
\begin{equation*}
\Lap\{f(t)\}(z) := \int_0^\infty f(t)e^{-zt}dt,
\end{equation*}
under certain technical conditions.
For example, since $\Lap\{t^n\}(z)=n!\,z^{-n-1}$, for all $n\in\Z_{\ge 0}$ and $\Re z>0$, then $\Lapformal$ coincides with $\Lap$ for polynomials $f(t)$ in the region $\Re z>0$.
In \cref{thm:r_transform}, we decided to use the formal $\Lapformal$ instead of $\Lap$ to avoid specifying any particular choice of analytic conditions for $f(t)$ and domains for $z$.
\end{remark}

\begin{remark}
Note that
\begin{equation}\label{eq:generating_bernoulli}
\frac{t}{1-e^{-t}} = \sum_{m \geq 0}\frac{(-1)^m B_m}{m!}t^m
\end{equation}
is the exponential generating function of the \emph{Bernoulli numbers}, the first few of them being: $B_0=1$, $B_1=-\frac{1}{2}$, $B_2=\frac{1}{6}$, $B_3=0$, $B_4=-\frac{1}{30}$, $B_5=0,\cdots$.
This shows that the relation between $c_1,c_2,\dots$ and $m_1,m_2,\dots$ from~\eqref{eq:2} is expressed in terms of these quantities.
The occurrence of the Bernoulli numbers in the theorem can be traced to an asymptotic expansion of the log gamma function in the proof below; see \cref{lem:help3}.
\end{remark}

\begin{remark}
Another reason for the interest in \cref{thm:r_transform} is the (possible) connection with Finite Free Probability, see e.g.~\cite{Marcusetal2022}.
In fact, similar identities were proved in the continuous settings of $\beta$-eigenvalues, see~\cite[Thm.~3.11]{Benaych-GeorgesCuencaGorin2022} and $\beta$-singular values, see \cite[Thm.~5.8]{Xu2023+}.
It is now known that the identities in those references are precisely the same as in the cumulant-based combinatorial approach to Finite Free Probability started in~\cite{ArizmendiPerales2018} (see also~\cite{ArizmendiGarciaPerales2023,Cuenca2024}), after the formal identification of the parameter $\gamma\mapsto\textrm{minus the degree of the polynomials}$.
Whether our equations~\eqref{eq:1}--\eqref{eq:2} have any relation to Finite Free Probability is an open question.
\end{remark}

\subsection{Preliminary lemmas}

We need a few lemmas for the proof of \cref{thm:r_transform}.
For the first two of them, we denote by $Q^*_\mu(x_1,\dots,x_N;\theta)$, $\mu\in\Y(N)$, the $Q$-normalization of the \emph{shifted Jack polynomials}, as defined in \cite{OkounkovOlshanski1997} or \cite[Secs.~2.2--2.4]{OkounkovOlshanski1998a}.
In particular, for row partitions $\mu=(k)$, they admit the explicit expression
\begin{equation}\label{eq:formulaQ}
Q_{(k)}^{*}(x_1,\dots,x_N;\theta) = \sum_{1\le i_1\le\cdots\le i_k\le N}{ \frac{\theta^{\uparrow m_1}\theta^{\uparrow m_2}\cdots}{m_1!m_2!\cdots} (x_{i_1}-k+1)\cdots (x_{i_{k-1}}-1)x_{i_k} },
\end{equation}
where $m_l:=\#\{ r\mid i_r=l \}$ are the multiplicities of $l=1,2,\dots$ in the sequence $i_1,\dots,i_k$.

\begin{lemma}\label{lem:help1}
If $\max_{1\le i\le N}{|x_i|}\le A$, for some $A>0$, then
\begin{equation}\label{lem:bound1}
\big| Q_{(k)}^{*}(x_1,\dots,x_N;\theta) \big|\le A^{\uparrow k}\cdot\sum_{1\le i_1\le\cdots\le i_k\le N}{ \frac{\theta^{\uparrow m_1}\theta^{\uparrow m_2}\cdots}{m_1!m_2!\cdots} },
\end{equation}
for all $k\in\Z_{\ge 1}$, where $m_l:=\#\{ r\mid i_r=l \}$ are the multiplicities of $l=1,2,\dots$ in the sequence $i_1,\dots,i_k$.
\end{lemma}
\begin{proof}
This follows immediately from~\eqref{eq:formulaQ}.
\end{proof}

\begin{lemma}\label{lem:help2}
If $\max_{1\le i\le N}{|x_i|}\le A$, for some $A>0$, then the series
\begin{equation}\label{sum_Qs}
\sum_{k\ge 0}{ \frac{(-1)^k Q_{(k)}^{*}(x_1,\dots,x_N;\theta)}{z^{\uparrow k}} }
\end{equation}
converges absolutely in the region $\Re z>N\theta+A$, and it is identically equal to
\begin{equation}\label{prod_gammas}
\prod_{i=1}^N{ \frac{\Gamma(x_i+z-i\theta)}{\Gamma(x_i+z-(i-1)\theta)}\frac{\Gamma(z-(i-1)\theta)}{\Gamma(z-i\theta)} }.
\end{equation}
\end{lemma}
\begin{proof}
This is proved in \cite[Appendix]{OkounkovOlshanski1998a}, so we only comment on the bound $\Re z>N\theta+A$, which is not explicitly stated in that reference.
Denote the RHS of~\eqref{lem:bound1} by $q_{A;\,\theta}^{(N)}(k)$, e.g.~$q_{A;\,\theta}^{(1)}(k) = A^{\uparrow k}\theta^{\uparrow k}/k!$.
To prove that~\eqref{sum_Qs} converges absolutely in the region $\Re z>N\theta+A$, it suffices to show by \cref{lem:help1} that the series
\begin{equation*}
F^{(N)}_{A;\,\theta}(y) := \sum_{k\ge 0}{ \frac{q_{A;\,\theta}^{(N)}(k)}{y^{\uparrow k}} }
\end{equation*}
converges if $y>N\theta+A$; in fact, by induction on $N$, we show it converges to $\frac{\Gamma(y-A-N\theta)\Gamma(y)}{\Gamma(y-A)\Gamma(y-N\theta)}$.

For the base case $N=1$, note that $F^{(1)}_{A;\,\theta}(y)$ is Gauss's hypergeometric series ${}_2F_1(\theta,A;y;1)$, for which it is well-known that it converges absolutely when $y>\theta+A$, and equals $\frac{\Gamma(y-A-\theta)\Gamma(y)}{\Gamma(y-A)\Gamma(y-\theta)}$, see e.g.~\cite[Sec.~1.3]{Bailey1935}.
For the inductive step, use the definitions and inductive hypothesis to deduce
\begin{multline}\label{eq:ind_step}
F_{A;\,\theta}^{(N)}(y) 
= \sum_{k\ge 0}{ F_{A+k;\,\theta}^{(N-1)}(y+k) \frac{q^{(1)}_{A;\,\theta}(k)}{y^{\uparrow k}} } 
= F_{A;\,\theta}^{(N-1)}(y) \sum_{k\ge 0}{ \frac{F_{A+k;\,\theta}^{(N-1)}(y+k)}{F_{A;\,\theta}^{(N-1)}(y)} \frac{q^{(1)}_{A;\,\theta}(k)}{y^{\uparrow k}} }\\
= F_{A;\,\theta}^{(N-1)}(y) \sum_{k\ge 0}{  \frac{\Gamma(y+k)\Gamma\big(y-(N-1)\theta\big)}{\Gamma\big(y+k-(N-1)\theta\big)\Gamma(y)} \frac{q^{(1)}_{A;\,\theta}(k)}{y^{\uparrow k}} }
= F_{A;\,\theta}^{(N-1)}(y) \sum_{k\ge 0}{  \frac{q^{(1)}_{A;\,\theta}(k)}{\big(y-(N-1)\theta\big)^{\uparrow k}} }.
\end{multline}
The last summation above is the hypergeometric series ${}_2F_1(\theta,A; y-(N-1)\theta;1)$, which converges if $y-(N-1)\theta > \theta+A$, and equals $\frac{\Gamma(y-A-N\theta)\Gamma(y-(N-1)\theta)}{\Gamma(y-(N-1)\theta-A)\Gamma(y-N\theta)}$ in that case.
Plugging back into~\eqref{eq:ind_step} and using the inductive hypothesis again concludes the induction.

A small refinement of this argument proves~\eqref{prod_gammas}; see \cite[Appendix]{OkounkovOlshanski1998a} for details.
\end{proof}

\begin{comment}
\begin{lemma}\label{lem:help3}
The following identity
\begin{equation}\label{eq:log_gamma}
\ln\Gamma(z) = \left(z-\frac{1}{2}\right)\ln z - z +\frac{\ln(2\pi)}{2}+\Lap\left\{ -\frac{1}{t^2}+\frac{1}{2t}+\frac{1}{t(e^t-1)}\right\}(z)
\end{equation}
holds whenever $\Re z>0$.
\end{lemma}
\begin{proof}
This result is known as Binet's first expression for $\ln\Gamma(z)$; see~\cite[Sec.~12.31]{WhittakerWatson2021} for a proof.
Let us only mention that $f(t)=-\frac{1}{t^2}+\frac{1}{2t}+\frac{1}{t(e^t-1)}$ is continuous at $t=0$, therefore $f(t)$ is bounded on $[0,\infty)$ and, hence, the Laplace transform $\Lap\{f(t)\}(z)$ is analytic on $\Re z>0$.
The lemma states that this analytic function equals $\ln\Gamma(z)-\left(z-\frac{1}{2}\right)\ln z + z -\frac{\ln(2\pi)}{2}$.
\end{proof}
\end{comment}

Recall that for a function $F(z)$ defined on the half-plane $|\arg z|<\pi/2$, the \emph{asymptotic expansion}
\begin{equation*}
F(z) \sim \sum_{m=0}^\infty{ b_m z^{-m} },\quad\text{as }z\to\infty,
\end{equation*}
means that
\begin{equation}\label{eq:asymptotic_expansion}
\lim_{\substack{|\arg(z)|<\pi/2 \\ z\to\infty}}{ z^N\left( F(z) - \sum_{m=0}^N{ b_m z^{-m} } \right) } = 0,\quad\textrm{for all }N\ge 0.
\end{equation}

\begin{lemma}\label{lem:help3}
The following is the asymptotic expansion of $\ln\Gamma(z)$, in the sense above,
\begin{equation}\label{eq:log_gamma}
\ln\Gamma(z) \sim \left(z-\frac{1}{2}\right)\ln z - z +\frac{\ln(2\pi)}{2} + \Lapformal\left\{ -\frac{1}{t^2}+\frac{1}{2t}+\frac{1}{t(e^t-1)}\right\}(z).
\end{equation}
\end{lemma}
\begin{proof}
By \cref{eq:generating_bernoulli} and the fact that $B_{2m+1}=0$, for all $m\ge 1$, we have
\[
-\frac{1}{t^2}+\frac{1}{2t}+\frac{1}{t(e^t-1)} = \sum_{m=1}^\infty{ \frac{B_{2m}}{2m(2m-1)}t^{2m-2} },
\]
in a neighborhood of zero, therefore
\[
\Lapformal\left\{ -\frac{1}{t^2}+\frac{1}{2t}+\frac{1}{t(e^t-1)}\right\}(z) = \sum_{m=1}^\infty{ \frac{B_{2m}}{2m(2m-1)}z^{1-2m} }.
\]
In terms of these Bernoulli numbers, the asymptotic expansion~\eqref{eq:log_gamma} is stated in~\cite[top of page 53]{Copson2004}.\footnote{There is a typo in the referenced equation from~\cite{Copson2004}: $\frac{1}{p^{2m}}$ should be $\frac{1}{p^{2m-1}}$.}
\end{proof}

\begin{lemma}\label{lem:help5}
For any formal power series $f(t)$ and $x\in\R$, we have
\begin{equation}\label{eq:identity_lap}
\Lapformal\{f(t)\}(z+x)=\Lapformal\{e^{-tx}f(t)\}(z).
\end{equation}
\end{lemma}
\begin{proof}
As mentioned in \cref{rem:real_lap_transform}, the formal Laplace transform and the usual Laplace transform agree on polynomials of arbitrary degree.
They also agree on functions of $t$ of the form $e^{-xt}\sum_{i=0}^N{a_it^i}$, for any $x>0$.
So for any polynomial $f(t)$ and $x>0$, we deduce
\begin{multline*}
\Lapformal\{f(t)\}(z+x) = \Lap\{f(t)\}(z+x) = \int_0^\infty{ f(t)e^{-(z+x)t}dt }\\
= \int_0^\infty{ \big( e^{-tx}f(t) \big) e^{-zt}dt } = \Lap\{e^{-tx}f(t)\}(z) = \Lapformal\{e^{-tx}f(t)\}(z),
\end{multline*}
i.e.~\eqref{eq:identity_lap} holds if $f(t)$ is a polynomial and $x>0$.
The desired \cref{eq:identity_lap} is an identity of formal power series in $z^{-1}$, meaning that we want to prove that the coefficients on both sides agree.
But such coefficients are just polynomials in $x$ and the coefficients of $f(t)$; since they agree for any $x>0$ and any polynomial $f(t)$, they must also agree for any $x\in\R$ and any formal power series $f(t)$, as needed.
\end{proof}

\begin{lemma}\label{lem:help4}
Let $\big\{L^{(N)}=\big(L^{(N)}_1,\dots,L^{(N)}_N\big)\big\}_{N\ge 1}$ be a sequence of tuples of real numbers such that
\begin{equation}\label{eq:limitM}
\lim_{N\to\infty}{ \frac{1}{N}\sum_{i=1}^N{ \Big(L^{(N)}_i\Big)^\ell } } = M_\ell,
\end{equation}
for some real numbers $M_1, M_2,\dots$ and all $\ell\in\Z_{\ge 1}$. Then
\[
\lim_{\substack{N\to\infty\\N\theta\to\gamma}}
\sum_{i=1}^N\left( \Big(L_i^{(N)}-\theta\Big)^\ell - \Big(L_i^{(N)}\Big)^\ell + \big((1-i)\theta\big)^\ell - (-i\theta)^\ell \right)
= (-1)^{\ell-1}\gamma^\ell - \ell\gamma M_{\ell-1},
\]
for all $\ell\in\Z_{\ge 1}$, where we set $M_0:=1$.
\end{lemma}
\begin{proof}
By a simple calculation, the prelimit expression can be written as
\begin{multline*}
\sum_{i=1}^N\left(\Big(L_i^{(N)}-\theta\Big)^\ell-\Big(L_i^{(N)}\Big)^\ell+(-\theta)^\ell\Big((i-1)^\ell-i^\ell\Big)\right)\\
= -\ell(\theta N)\cdot\frac{1}{N}\sum_{i=1}^N \left(L_i^{(N)}\right)^{\ell-1}
+ \sum_{k=2}^\ell\sum_{i=1}^N\binom{\ell}{k}\left(L_i^{(N)}\right)^{\ell-k}(-\theta)^{k} - (-\theta)^{\ell}N^\ell.
\end{multline*}
The first term above converges to $-\ell\gamma M_{\ell-1}$, in the regime where $N\to\infty$, $N\theta\to\gamma$, due to the assumption~\eqref{eq:limitM}.
The same relation shows that the second term converges to zero.
The last term evidently converges to $(-1)^{\ell-1}\gamma^\ell$, finishing the proof.
\end{proof}

We are finally ready to prove \cref{thm:r_transform}.
It should be noted that the main tool will be \cref{theo:main1}, which in turn was proved by employing Jack generating functions and Cherednik operators.
It would be desirable to have an elementary proof along the same lines as \cite[Main Result~I]{Cuenca2024}.
We leave finding a combinatorial proof as an open problem.

\subsection{Proof of Theorem~\ref{thm:r_transform}}

Let $\mu$ be any compactly supported probability measure on $\R$ for which there exists a sequence of partitions $\{ \la^{(N)}\in\Y(N) \}_{N\ge 1}$ such that $C\ge\la^{(N)}_1\ge\cdots\ge\la^{(N)}_N\ge 0$, for some constant $C>0$, and such that we have the following weak limit
\begin{equation*}
\frac{1}{N}\sum_{i=1}^N{ \delta_{\LL_i^{(N)}} }\to\mu,\quad\text{as }N\to\infty,\ N\theta\to\gamma,
\end{equation*}
where $\LL_i^{(N)}=\la_i^{(N)}-\theta(i-1)$, for all $1\le i\le N$.
As $\mu$ is compactly supported, it is implied that the sequence of probability measures
\begin{equation}\label{eq:empirical_proof}
\PP_N:=\frac{1}{N}\sum_{i=1}^N{ \delta_{\LL_i^{(N)}} }
\end{equation}
satisfies the LLN, in the sense of \cref{def:LLN}.
In particular, if we let $m_1,m_2,\cdots$ be the moments of $\mu$, then
\begin{equation}\label{eq:moment_limits}
\lim_{\substack{N\to\infty\\N\theta\to\gamma}}{ \frac{1}{N}\sum_{i=1}^N{ \Big( \LL_i^{(N)} \Big)^k } } = m_k,\quad\text{for all }k\in\Z_{\ge 1}.
\end{equation}
The Jack generating function of $\PP_N$ is equal to
\begin{align}
G_{N,\theta}(x_1+1,\dots,x_n+1) &= \frac{P_{\la^{(N)}}(x_1+1,\dots,x_N+1;\theta)}{P_{\la^{(N)}}(1^N;\theta)}\nonumber\\
&= \sum_{\mu\in\Y(N)}\frac{Q_\mu^{*}\Big(\la_1^{(N)},\dots,\la_N^{(N)};\theta\Big) P_\mu\big(x_1,\dots,x_N;\theta\big)}{(N\theta)_\mu},\label{eqn:JGF}
\end{align}
where the second equality above is proved in~\cite[Thm.~(3.2)]{OkounkovOlshanski1997}.
In~\eqref{eqn:JGF}, the symbol in the denominator is $(g)_\mu:=\prod_{(i,j)\in\mu}\big( g+(j-1)-\theta(i-1) \big)$.

We want to specialize the variables $x_1 = z$ and $x_2 = \cdots = x_N = 0$ in~\eqref{eqn:JGF}.
Using the facts (i)~$P_\mu(z,0^{N-1};\theta)=0$, if $\ell(\mu)>1$, (ii)~$P_{(k)}(z,0^{N-1};\theta)=z^k$, for all $k\in\Z_{\ge 0}$, (iii)~$Q_{(0)}^*(x_1,\dots,x_N;\theta)=1$, and (iv)~$(N\theta)_{(k)}=(N\theta)^{\uparrow k}$, for all $k\in\Z_{\ge 0}$, we deduce:
\begin{equation}\label{eq:expansion_G}
G_{N,\theta}\big(z+1,\,1^{N-1}\big) = 1+\sum_{k\ge 1}\frac{Q_{(k)}^*\Big(\la_1^{(N)},\dots,\la_N^{(N)};\theta\Big)}{(N\theta)^{\uparrow k}}z^k.
\end{equation}
Since $\{\PP_N\}_{N\ge 1}$ satisfies the LLN, \cref{theo:main1} implies HT-appropriateness. In particular, we have the following convergence of Taylor coefficients of the logarithms:
\begin{equation}\label{eq:expansion_G_2}
\lim_{\substack{N\to\infty\\ N\theta\to\gamma}}{ \frac{1}{(n-1)!}\frac{d^n}{dz^n} \ln\big(G_{N,\theta}\big(z+1,\,1^{N-1}\big)\big) \big|_{z=0} } = \kappa_n,\quad n\in\Z_{\ge 1}.
\end{equation}
This has two implications.
First, the functions $G_{N,\theta}(z+1,1^{N-1})$ themselves have converging Taylor coefficients, so due to the expansion~\eqref{eq:expansion_G}, we deduce that the following limits exist:
\begin{equation}\label{eq:limit_ck}
c_k = \lim_{\substack{N\to\infty\\N\theta\to\gamma}}Q_{(k)}^{*}\Big( \la_1^{(N)},\dots,\la_N^{(N)}; \theta \Big),
\quad \text{for all }k\in\Z_{\ge 1}.
\end{equation}
Second, by \eqref{eq:expansion_G} and \eqref{eq:expansion_G_2}, we actually deduce the following equality
\begin{equation*}
\exp\left(\sum_{n\geq 1}\frac{\kappa_n z^n}{n}\right) = 1+\sum_{k \geq 1}\frac{c_k}{\gamma^{\uparrow k}}z^k,
\end{equation*}
that is to be interpreted coefficient-wise, i.e.~as an equality between formal power series.
This is precisely the first equation~\eqref{eq:1} in the statement of the theorem.
We have thus proved that the quantities $c_1,c_2,\dots$, uniquely determined from $\kappa_1,\kappa_2,\dots$ by this equation, are exactly the limits~\eqref{eq:limit_ck} of certain evaluations of shifted Jack polynomials.

\smallskip

On the other hand, according to \cref{lem:help4}, we have the following generating function in $z^{-1}$ for shifted Jack polynomials:
\begin{equation}\label{eq:2half}
1+\sum_{k\ge 1}{ \frac{(-1)^k Q_{(k)}^{*}(x_1,\dots,x_N;\theta)}{z^{\uparrow k}} }
= \prod_{i=1}^N{ \frac{\Gamma(x_i+z-i\theta)}{\Gamma(x_i+z-(i-1)\theta)}\frac{\Gamma(z-(i-1)\theta)}{\Gamma(z-i\theta)} }.
\end{equation}
If we set $x_i=\la_i^{(N)}$, for $i=1,\dots,N$, then~\eqref{eq:2half} is, in fact, an equality between analytic functions in the region $\Re z > N\theta+C$, as shown in \cref{lem:help2}.
We will soon be concerned with the high temperature regime where $N\to\infty$, $N\theta\to\gamma$; in this regime, the identity~\eqref{eq:2half} holds whenever $\Re z\gg 0$, i.e., in a region $\{z\in\C\colon\Re z > C'\}$, for some large enough $C'>0$.

By taking the logarithm of~\eqref{eq:2half}, we have
\begin{equation}\label{eq:3}
\ln\left( 1+\sum_{k\ge 1}\frac{(-1)^k Q_{(k)}^*\Big( \la_1^{(N)},\dots,\la_N^{(N)};\theta \Big)}{z^{\uparrow k}} \right) = \sum_{i=1}^N \sum_{j=1}^4{ (-1)^{j-1} \ln\Gamma\left(z+\nabla_{i;j}^{(N)}\right) },
\end{equation}
where
\begin{equation*}
\nabla_{i;1}^{(N)} := \LL_i^{(N)}-\theta,\quad \nabla_{i;2}^{(N)} := \LL_i^{(N)},
\quad \nabla_{i;3}^{(N)} := (1-i)\theta,\quad \nabla_{i;4}^{(N)} := -i\theta,
\end{equation*}
and \eqref{eq:3} is valid whenever $\Re z\gg 0$.
Next, we want to consider the asymptotic expansion of both sides in the region $|\arg z|<\pi/2$, as $z\to\infty$, in the sense of \cref{eq:asymptotic_expansion}.
The LHS is the same because it is an analytic function, while the asymptotic expansion of the RHS can be deduced from \cref{lem:help3,lem:help5}.
We end up with the following equality, which should be interpreted as an equality between formal power series in $z^{-1}$:
\begin{align}
\ln\left( 1+\sum_{k\ge 1}\frac{(-1)^k Q_{(k)}^*\Big( \la_1^{(N)},\dots,\la_N^{(N)};\theta \Big)}{z^{\uparrow k}} \right) = \sum_{i=1}^N\sum_{j=1}^4(-1)^{j-1}\left(z+\nabla_{i;j}^{(N)}-\frac{1}{2}\right)\ln\left( 1+\frac{\nabla_{i;j}^{(N)}}{z} \right)\nonumber\\ 
+ \Lapformal\left\{\sum_{i=1}^N\sum_{j=1}^4(-1)^{j-1}e^{-t\nabla_{i;j}^{(N)}}\left(-\frac{1}{t^2}+\frac{1}{2t}+\frac{1}{t(e^t-1)}\right)\right\}(z).  \label{eq:4}
\end{align}
We can further simplify this expression by denoting
\begin{equation*}
\widetilde{\nabla}_\ell^{(N)} := \sum_{i=1}^N \sum_{j=1}^4{ (-1)^{j-1}\Big(\nabla_{i;j}^{(N)}\Big)^\ell },
\quad \ell\in\Z_{\ge 1}.
\end{equation*}
For example, note that $\widetilde{\nabla}_1^{(N)}=0$.
Then the RHS of~\eqref{eq:4} can be rewritten as
\begin{equation}\label{eq:5}
\sum_{\ell\ge 1}(-1)^\ell\left(\frac{\widetilde{\nabla}^{(N)}_{\ell+1}}{\ell+1}-\frac{\widetilde{\nabla}^{(N)}_{\ell+1}}{\ell}+\frac{\widetilde{\nabla}^{(N)}_{\ell}}{2\ell}\right)z^{-\ell} + \Lapformal\left\{\sum_{\ell\ge 1} (-t)^\ell\frac{\widetilde{\nabla}^{(N)}_\ell}{\ell!}\left(-\frac{1}{t^2}+\frac{1}{2t}+\frac{1}{t(e^t-1)}\right)\right\}(z).  
\end{equation}
Next, we want to take the coefficient-wise limit of the formal power series \eqref{eq:4}--\eqref{eq:5} in the regime $N\to\infty$, $N\theta\to\gamma$.
We will use~\eqref{eq:moment_limits} and \cref{lem:help4}, which imply
\begin{equation}\label{eq:6}
\lim_{\substack{N\to\infty\\N\theta\to\gamma}}{ \widetilde{\nabla}_\ell^{(N)} } = (-1)^{\ell-1}\gamma^\ell - \ell\gamma m_{\ell-1},
\end{equation}
for all $\ell\in\Z_{\ge 1}$, if we set $m_0:=1$.
As a result, the limit of the first term in~\eqref{eq:5} equals
\begin{multline}\label{eq:contribution_1}
\gamma \sum_{\ell\geq 1}\left(\frac{(-1)^\ell m_{\ell}}{\ell}+\frac{(-1)^{\ell-1}
m_{\ell-1}}{2} + \frac{\gamma^\ell}{\ell+1} - \frac{\gamma^\ell}{\ell} - \frac{\gamma^{\ell-1}}{2\ell}\right)z^{-\ell}\\
= \Lapformal\left\{\left(\gamma\cdot  M(-t) -\frac{e^{t\gamma}-1}{t}\right)\left(\frac{1}{t}+\frac{1}{2}\right)\right\}(z),
\end{multline}
where $M(s):=1+\sum_{n\ge 1}{\frac{m_n}{n!}s^n}$ is the moment generating function of $\mu$.
Similarly, due to~\eqref{eq:6}, the limit of the second term of~\eqref{eq:5} is equal to
\begin{equation}\label{eq:contribution_2}
\Lapformal\left\{\left(\gamma\cdot  M(-t) - \frac{e^{t\gamma}-1}{t}\right)\left(-\frac{1}{t}+\frac{1}{2}+\frac{1}{e^t-1}\right)\right\}(z).
\end{equation}
We conclude that the sum of \eqref{eq:contribution_1} and \eqref{eq:contribution_2} is the limit of~\eqref{eq:5} (the RHS of~\eqref{eq:4}).
On the other hand, the limit of the LHS of~\eqref{eq:4} follows from~\eqref{eq:limit_ck}.
By putting both of these results together, we obtain the identity
\begin{equation*}
\ln\left(1+\sum_{n \ge 1}\frac{(-1)^nc_n}{z^{\uparrow n}}\right)
= \Lapformal\left\{\left(\gamma\cdot  M(-t) - \frac{e^{t\gamma}-1}{t}\right)\frac{1}{1-e^{-t}}\right\}(z),
\end{equation*}
as the limit of~\eqref{eq:4} in the regime $N\to\infty$, $N\theta\to\gamma$; this is precisely the desired~\eqref{eq:2}.

In order to finish the proof, we need to conclude that the relation between generating functions holds true for arbitrary sequences $m_k$, and not only for those of the form~\eqref{eq:moment_limits}. Indeed, note that the two equations \eqref{eq:1}--\eqref{eq:2} lead to an expression for $\kappa_n$ as a polynomial in $m_1,\dots,m_n$ over $\C(\gamma)$. On the other hand, \eqref{eq:kappa-top} also leads to an expression of $\kappa_n$ as a polynomial in $m_1,\dots,m_n$ over $\C(\gamma)$.
Fix $n\in\Z_{\ge 1}$ and denote these polynomials by $\tilde{P}$ and $\tilde{P}'$.
There exists a polynomial $f\in\C[\gamma]$ such that both $P := f \cdot \tilde{P}$ and $P' := f\cdot\tilde{P}'$ are polynomials in $\gamma,m_1,\dots,m_n$ over $\C$. The proof so far has shown that $P$ and $P'$ agree when $\gamma\in(0,\infty)$ and the $m_k$'s are of the form~\eqref{eq:moment_limits}.
We claim that $P$ and $P'$ are in fact identical, whence $\tilde{P} = \tilde{P}'$, and the proof would be finished.

Fix an arbitrary partition $\la=(\la_1,\dots,\la_n)\in\Y(n)$ and define $\la_n^{(N)}\in\Y(N)$ as
\[
\la_n^{(N)} :=
(\underbrace{\la_1,\dots,\la_1}_{\frac{N}{n}+O(1) \text{ terms}},\dots,
\underbrace{\la_n,\dots,\la_n}_{\frac{N}{n}+O(1) \text{ terms}}).
\]
It is easy to show that the empirical measures~\eqref{eq:empirical_proof} corresponding to the sequence $\big\{ \la_n^{(N)} \big\}_{N\ge 1}$ satisfy the LLN in the sense of \cref{def:LLN}, with corresponding $m_1,m_2,\cdots\in\R$ given by
\begin{multline*}
m_k := \lim_{\substack{N\to\infty\\N\theta\to\gamma}}{\frac{1}{N}\sum_{j=1}^N{ \Big( \LL_j^{(N)} \Big)^k } } 
= \lim_{\substack{N\to\infty\\N\theta\to\gamma}}\frac{1}{N}\sum_{j=1}^n\sum_{k=0}^{\frac{N}{n}+O(1)} \left(\la_j-\theta \left( (j-1)\frac{N}{n}+k+O(1)\right)
\right)^k\\
= \sum_{j=1}^n\int_{\frac{j-1}{n}}^{\frac{j}{n}}(\la_j-\gamma x)^idx
= \frac{1}{n(k+1)}
\sum_{j=1}^n \sum_{s=0}^k \binom{k+1}{s} \left(\la_j-\frac{j\gamma}{n}\right)^s\left(\frac{\gamma}{n}\right)^{k-s}.
\end{multline*}
As a result, $P(\gamma,m_1,\dots,m_n) = P'(\gamma,m_1,\dots,m_n)$ on the following subset of $\C^{n+1}$:
\begin{multline*}
X :=\{(x_0,x_1,\dots,x_n)\colon x_0 \in (0,\infty), \\
x_k = \frac{1}{n(k+1)}
\sum_{j=1}^n \sum_{s=0}^k \binom{k+1}{s} \left(\la_j-\frac{j\gamma}{n}\right)^s\left(\frac{\gamma}{n}\right)^{k-s},\text{ for }1\le k\le n,\ (\la_1,\dots,\la_n)\in\Y(n)\bigg\}.
\end{multline*}
Notice now that $X$ is the image of $(0,\infty)\times\Y(n)\subset\C^{n+1}$ through the polynomial map $F = (f_0,f_1,\dots,f_n) \in \C[x_0,x_1,\dots,x_n]^{n+1}$ defined by
\[
f_0 := x_0,\qquad f_k := \frac{1}{n(k+1)}
\sum_{j=1}^n \sum_{s=0}^k \binom{k+1}{s} \left(x_j-\frac{jx_0}{n}\right)^s\left(\frac{x_0}{n}\right)^{k-s}, \text{ for }1\le k\le n.
\]
Observe that $f_0,f_1,\dots,f_n \in \C[x_0,x_1,\dots,x_n]$ are algebraically independent, therefore the map $F$ is dominant (its image is dense in the Zariski topology), so it maps Zariski-dense subsets of $\C^{n+1}$ into Zariski-dense subsets. Thus, $X$ is Zariski-dense, since it is the image of the Zariski-dense subset $(0,\infty)\times\Y(n)\subset\C^{n+1}$. Consequently, $P$ and $P'$ are identical, for being polynomial functions that coincide on $X$; this concludes the proof.

\section{Characteristic functions of the alpha, beta and Plancherel Jacobi operators}
\label{sec:Jacobi}

We are interested in determining the spectra of three families of Jacobi operators.
Two of them are associated to infinite Jacobi matrices, or operators on $\ell^2(\Z_{\ge 1})$, and one is associated to a finite Jacobi matrix.
We will show that the zeroes of certain entire functions, based on generalized hypergeometric functions, coincide with the spectra of our Jacobi operators.
These entire functions are called \emph{characteristic functions}; in the case of the finite Jacobi matrix, this will be a constant multiple of the characteristic polynomial.
The asymptotic analysis of eigenvalues of our Jacobi operators, as well as the problem of computing the characteristic functions, have already been studied; see, e.g.~\cite{JanasNaboko2001,JanasMalejki2006,BoutetdeMonvelZielinski2008,StampachStovicek2015} and references therein.
Let us review the relevant results needed for our applications.

\smallskip

The Jacobi matrices $J_{\la;w}$ in question are of the form
\[  J_{\la;w} := \begin{bmatrix}
\lambda_1 & w_1 & 0 & 0 & \cdots\\
w_1 & \lambda_2 & w_2 & 0 & \ddots\\
0 & w_2 & \lambda_3 & w_3 & \ddots\\
0 & 0 & w_3 & \lambda_4 & \ddots\\
\vdots & \ddots & \ddots & \ddots & \ddots
\end{bmatrix}, \]
where rows and columns are parametrized by $\Z_{\ge 1}$, or by $\{1,2,\dots,M\}$ in the finite case.
The quantities $\{w_i\}_{i\ge 1}\subset\R^*$ and $\{\la_i\}_{i\ge 1}\subset\R$ are called the Jacobi parameters; we assume that all $w_i$'s have the same sign.
In the infinite cases, we require $\sum_{i\ge 1}|w_i|^{-1} = \infty$, so that $J_{\la;w}$ is an essentially self-adjoint operator (on the space of finite vectors) and therefore admits a unique self-adjoint extension to $\ell^2(\Z_{\geq 1})$.

The two infinite Jacobi matrices of our interest fit into the following two general classes:
\begin{equation}\label{class_1}
\la_i = \zeta-i,\quad |w_i| = O(i^{1-\epsilon}),\qquad\zeta\in\R,\quad\epsilon>0,
\end{equation}
and
\begin{equation}\label{class_2}
\la_i = \zeta-\delta i,\quad w_i = \sqrt{\alpha i^2+\beta i+\gamma},
\quad\ \zeta\in\R,\quad\alpha,\delta > 0,\quad\beta^2\ge 4\alpha\gamma,\quad\delta > 2\sqrt{\alpha}.
\end{equation}
In both cases, it is known that the spectrum of $J_{\la;w}$ is discrete, bounded from above, and the eigenvalues are simple (see \cite[Thm.~5.1]{JanasNaboko2001}, \cite[Prop.~1]{BoutetdeMonvelZielinski2008}, and \cite[Sec.~3]{StampachStovicek2015}).

Let us specialize to three multiparameter families of Jacobi matrices, that will be labeled by ``alpha'', ``beta'' and ``planch'', in connection to \cref{thm:KOO}.
The Jacobi matrices $J^{(\text{alpha})}$ and $J^{(\text{planch})}$ will be infinite, while $J^{(\text{beta})}$ will be finite.

\smallskip

Let $\gamma,\eta > 0$, $c \in (0,1)$, $M\in\Z_{\ge 1}$.
The Jacobi matrices $J^{(\text{alpha})},J^{(\text{beta})},J^{(\text{planch})}$ are defined by the following choices of Jacobi parameters,
\begin{align}
\la^{(\text{alpha})}_i &:= -i\cdot\frac{1+c}{1-c} - \left(\gamma+\eta+\frac{1}{c}\right)\frac{c}{1-c},\quad 
w^{(\text{alpha})}_i := \sqrt{\frac{(i+\gamma)(i+\eta)}{c}}\frac{c}{1-c},\label{eq:Jacobialpha}\\
\la^{(\text{beta})}_i &:= M(c-1)+c(\gamma +1)+i(1-2c),\quad w^{(\text{beta})}_i := \sqrt{c(1-c)i(\gamma + M-i)},\label{eq:Jacobibeta}\\
\la^{(\text{planch})}_i &:= 1-\eta-i, \quad w^{(\text{planch})}_i := -\sqrt{\eta(\gamma+i)},\label{eq:Jacobiplanch}
\end{align}
where in \eqref{eq:Jacobibeta}, we have $1\le i\le M$ for the indices of $\la_i^{(\text{beta})}$ and $1\le i\le M-1$ for the indices of $w_i^{(\text{beta})}$, so that $J^{(\text{beta})}$ is a finite $M\times M$ matrix, while for equations \eqref{eq:Jacobialpha} and \eqref{eq:Jacobiplanch}, we have $i\ge 1$, so that $J^{(\text{alpha})}$ and $J^{(\text{planch})}$ are infinite $\Z_{\ge 1}\times\Z_{\ge 1}$ matrices.
Note that $J^{(\text{alpha})},J^{(\text{beta})},J^{(\text{planch})}$ depend on $\gamma,\eta,c,M$, but this dependence will be hidden from the notation, for convenience.

Observe that $J^{(\text{alpha})}$ and $J^{(\text{planch})}$ fit into the framework of equations \eqref{class_2} and \eqref{class_1}, respectively; in particular, their spectra are discrete, bounded from above, and the eigenvalues are simple, therefore they can be arranged in decreasing order and labeled by $\Z_{\ge 1}$.

\begin{lemma}\label{lem:roots_lemma}
Label the eigenvalues of $J^{(\text{alpha})}$ and $J^{(\text{planch})}$ as follows:
\[
\spec\left(J^{(\text{alpha})}\right) = \big\{\ell^{(\text{alpha})}_1 > \ell^{(\text{alpha})}_2 > \dots\big\},\quad
\spec\left(J^{(\text{planch})}\right) = \big\{\ell^{(\text{planch})}_1 > \ell^{(\text{planch})}_2 > \dots\big\}.
\]
Then $\ell^{(\text{planch})}_n = 1-n+O(n^{-M})$, for any $M>0$, as $n\to\infty$.
Moreover, if $\gamma=\eta$, then also $\ell^{(\text{alpha})}_n = -n+O(n^{-M})$, for any $M>0$, as $n\to\infty$. In the generic case $\gamma\ne\eta$, we have that $\ell^{(\text{alpha})}_n = -n+O(n^{-1})$, as $n\to\infty$.\footnote{It is possible that, in the generic case $\gamma\ne\eta$, we also have $\ell^{(\text{alpha})}_n = -n+O(n^{-M})$, for any $M>0$, as $n\to\infty$, but we do not know how to prove it.}
\end{lemma}
\begin{proof}
For the Plancherel case, first consider the Jacobi matrix $J_{\la,w}$ with $\la_i=i$ and $w_i=\sqrt{\eta(\gamma+i)}$. Then \cite[Theorem 2]{BoutetdeMonvelZielinski2008} can be applied with parameters $\rho=\rho'=1/2$, $\rho_1=M+1/2$, and $\gamma_n=-\eta$, for $n>1$, and gives that $\ell_n(J_{\la,w})=n-\eta+O(n^{-M})$.
The result follows from noticing that $J^{(\text{planch})}=(1-\eta)I-J_{\la,w}$, and therefore $\ell_n^{(\text{planch})}=1-\eta-\ell_n(J_{\la,w})$.

The statement about the eigenvalues $\ell^{(\text{alpha})}_n$ follows from~\cite[Sections 3 and 4]{JanasMalejki2006}.\footnote{That paper studied the case when the parameters in \eqref{eq:Jacobiplanch} are negated, but the proof is the same in our case.}
\end{proof}

To state the main result of this subsection, we recall some terminology.
Let $a_1,a_2,a,b\in\C$.
The \emph{(Gauss's) hypergeometric function} ${}_2F_1(a_1,a_2;b;z)$ and \emph{confluent hypergeometric function} ${}_1F_1(a;b;z)$ are the series
\begin{equation}\label{eq:series_def}
{}_2F_1(a_1,a_2;b;z) = \sum_{n\ge 0}\frac{a_1^{\uparrow n} a_2^{\uparrow n}}{b^{\uparrow n}}\cdot\frac{z^n}{n!},\qquad
{}_1F_1(a;b;z) = \sum_{n\ge 0}\frac{a^{\uparrow n}}{b^{\uparrow n}}\cdot\frac{z^n}{n!}.
\end{equation}
We will regard them as functions of the complex variables $b,z$.
It is known that $\frac{1}{\Gamma(b)}\cdot{}_1F_1(a_1,a_2;b;z)$ is holomorphic on $b,z$, while $\frac{1}{\Gamma(b)}\cdot{}_2F_1(a_1,a_2;b;z)$ is holomorphic on $b$ and on $z$ in the unit disk $|z|<1$.
If $a_1=-M$ or $a_2=-M$, for some $M\in\Z_{\ge 1}$, then $\frac{\Gamma(b+M)}{\Gamma(b)}\cdot{}_2F_1(a_1,a_2;b;z)$ is a polynomial in $b$ and $z$.

\begin{lemma}\label{lem:pfaff}
We have
\begin{equation*}
{}_2F_1\left(A,B;z;\frac{C}{C-1}\right) = (1-C)^A\cdot{}_2F_1\left(A,z-B;z;C\right) = (1-C)^B\cdot{}_2F_1\left(z-A,B;z;C\right),
\end{equation*}
for any $A,B,C\in\C$, where both sides are interpreted as formal power series in $z^{-1}$, and the hypergeometric functions ${}_2F_1$ are the power series in \eqref{eq:series_def}.
\end{lemma}
\begin{proof}
When $A,B,C\in\C$ and $|C|<1/2$, these equalities of analytic functions of $z$ on $\C\setminus\Z_{\ge 0}$ are known as Pfaff's transformations.
Therefore, the equalities also holds as equalities of formal power series.
\end{proof}

\begin{lemma}\label{lem:contiguous}
The following is an equality of analytic functions, where $z$ is on the unit disk:
\begin{multline*}
\left( (C-1) + (B-A)z \right)\cdot{}_2F_1\left(A,B;C;z\right) = C^{-1}(C-A)Bz\cdot{}_2F_1\left(A,B+1;C+1;z\right)\\
+ (C-1)\cdot{}_2F_1\left(A,B-1;C-1;z\right).
\end{multline*}
\end{lemma}
\begin{proof}
Three-term relations of this kind are known as \emph{Gauss's contiguous relations}.
The one in this lemma, in particular, can be deduced from \cite[Eqns.~(1.5) and (2.19)]{RakhaRathieChopra2011}.
\end{proof}

\begin{theorem}\label{thm:char_fns}
Let $\gamma,\eta > 0$, $c \in (0,1)$, $M\in\Z_{\ge 1}$.
The following functions
\begin{align}
F^{\text{alpha}}(z) &:= \frac{(1-c)^\gamma}{\Gamma(z)}\cdot{}_2F_1\left(\gamma,z-\eta;z;c\right) = \frac{(1-c)^\eta}{\Gamma(z)}\cdot{}_2F_1\left(z-\gamma,\eta;z;c\right),\label{eq:Zeroes_alpha}\\
F^{\text{beta}}(z) &:= \frac{\Gamma(z+M)}{\Gamma(z)}\cdot{}_2F_1\left(-M,\gamma; z; c\right), \label{eq:Zeroes_beta}\\
F^{\text{planch}}(z) &:= \frac{1}{\Gamma(z)}\cdot{}_1F_1\left(\gamma;z;-\eta\right),\label{eq:Zeroes_planch}
\end{align}
are entire functions and, in particular, $F^{\text{beta}}(z)$ is a monic polynomial of degree $M$.
Moreover, the spectrum (set of eigenvalues) of $J^{(\text{alpha})}$ (resp. $J^{(\text{beta})}$, $J^{(\text{planch})}$) coincides with the set of roots of $F^{\text{alpha}}(z)$ (resp. $F^{\text{beta}}(z)$, $F^{\text{planch}}(z)$).
As a result, the zeroes of the functions in \eqref{eq:Zeroes_alpha}, \eqref{eq:Zeroes_beta}, \eqref{eq:Zeroes_planch} are all real.
\end{theorem}

\begin{proof}
\emph{Plancherel case.}
It is shown in \cite[Proposition 8]{StampachStovicek2015} that the eigenvalues of the Jacobi operator with parameters $\la_i = fi$, $w_i = \sqrt{d+ei}$, where $e,f>0$, $d+e>0$, coincides with the zeroes of the function
\[{}_1F_1\left(1-\frac{d}{e}-\frac{e}{f^2}-\frac{z}{f};1-\frac{e}{f^2}-\frac{z}{f};\frac{e}{f^2}\right)\cdot \Gamma\left(1-\frac{e}{f^2}-\frac{z}{f}\right)^{-1}.\] 
By choosing $d=\eta\gamma$, $e=\eta$ and $f=1$, we obtain that the eigenvalues of the Jacobi operator with parameters $\la_i = i$, $w_i = \sqrt{\eta(\gamma+i)}$, coincides with the zeroes of ${}_1F_1\left(1-\gamma-\eta-z;1-\eta-z;\eta\right)\cdot \Gamma\left(1-\eta-z\right)^{-1}$. By using the Kummer's transformation
\[
{}_1F_1\left(1-\gamma-\eta-z;1-\eta-z;\eta\right) = e^\eta\cdot{}_1F_1\left(\gamma;1-\eta-z;-\eta\right),
\]
then noticing that negating the Jacobi operator and adding $(1-\eta)I$ is equivalent to shifting the argument $z\mapsto 1-\eta-z$ of the characteristic function, we conclude that the eigenvalues of $J^{(\text{planch})}$ coincide with the zeroes of $F^{\text{planch}}(z)={}_1F_1\left(\gamma;z;-\eta\right)\cdot \Gamma(z)^{-1}$, as desired.

\smallskip

\emph{Beta case.}
In this case, we prove that $F^{\text{beta}}(z)=\frac{\Gamma(z+M)}{\Gamma(z)}\cdot{}_2F_1\left(-M,\gamma; z; c\right)$ is the characteristic polynomial of $J^{(\text{planch})}$.
We can rewrite this function as $F^{\text{beta}}(z)=\gamma^{\uparrow M}\cdot c^M\cdot{}_2F_1\left(-M,1-M-z;1-M-\gamma;1/c\right)$.
From this formula, it is easy to see that $\frac{1}{\gamma^{\uparrow M}\cdot c^M}\cdot F^{\text{beta}}(z-M+1)$ coincides with (a normalized version of) the classical \emph{Krawtchouk orthogonal polynomial} $K_n(x;p,N)$, as given by \cite[Eqn.~(1.10.1)]{KoekoekSwarttouw1996}, upon the formal identification of parameters $c=p$, $M=n$, $\gamma=N+1-n$, $z=x$.
Therefore, it must satisfy the three-term recurrence relation of the Krawtchouk polynomials, namely \cite[Eqn.~(1.10.3)]{KoekoekSwarttouw1996}, since both sides are rational functions of all parameters involved.
The recurrence relation will be more clear if we highlight the role of the parameters $\gamma, M$, and define first $G_{\gamma,M}(z) := F^{\text{beta}}(z-M+1)$; the relation is then
\begin{equation}\label{eq:G_relation}
G_{\gamma-1,M+1}(z) = \big( x - M(1-c) - c(\gamma-1) \big)G_{\gamma,M}(z) - M\gamma c(1-c)G_{\gamma+1,M-1}(z).
\end{equation}
Note that the parameter $\gamma$ is different for all three functions above.
Then define $H_M(z) := G_{\gamma-M,M}(z)$, so that after setting $\gamma\mapsto\gamma-M$ in equation \eqref{eq:G_relation}, we can rewrite it as
\begin{equation}\label{eq:H_relation}
H_{M+1}(z) = \big( z-M(1-c)-c(\gamma-M-1) \big)H_M(z) - M(\gamma-M) c(1-c)H_{M-1}(z).
\end{equation}
We claim that $H_M(z)$ is the characteristic polynomial of the $M\times M$ Jacobi matrix $J_{\la',w'}$ with parameters $\la_i'=(i-1)(1-c)+c(\gamma-i)$, for $1\le i\le M$, and $w_i'=\sqrt{i(\gamma-i)c(1-c)}$, for $1\le i\le M-1$.
Indeed, this follows from the observation that $\det\big[z\cdot I_M - J_{\widetilde{\la},\widetilde{w}}\big]$ can be calculated by expanding the determinant along the final row, resulting exactly in the recurrence relation \eqref{eq:H_relation}; since $H_M(z)$ is monic, it must coincide with the determinant.
Then after replacing $\gamma$ by $\gamma+M$, we obtain that $G_{\gamma,M}(z)$ is the characteristic polynomial of the $M\times M$ Jacobi matrix $J_{\la'',w''}$ with parameters $\la_i''=(i-1)(1-c)+c(\gamma+M-i)$, $w_i''=\sqrt{i(\gamma+M-i)c(1-c)}$.
It follows that $F^{\text{beta}}(z)=G_{\gamma,M}(z+M-1)$ is the characteristic polynomial of the $M\times M$ Jacobi matrix with parameters $\la_i=\la_i''-M+1$, $w_i=w_i''$, which is exactly $J^{(\text{beta})}$; this ends the proof.

\smallskip

\emph{Alpha case.}
Define the functions
\begin{equation*}
v_i(z) := \sqrt{\frac{\Gamma(\gamma+i)\Gamma(\eta+i)c^i}{\Gamma(\gamma)\Gamma(\eta)}}\cdot\frac{(1-c)^{\frac{z}{1-c}-\gamma}}{\Gamma\left( \frac{z}{1-c} + i \right)}\,{}_2F_1\left(\frac{z}{1-c}-\gamma,\,\eta+i;\,\frac{z}{1-c}+i;\,c\right),
\end{equation*}
for all $i\in\Z_{\ge 0}$.
Note that 
\[
v_0(z) = \frac{(1-c)^{\frac{z}{1-c}-\gamma}}{\Gamma\left( \frac{z}{1-c} \right)}\,{}_2F_1\left(\frac{z}{1-c}-\gamma,\,\eta;\,\frac{z}{1-c};\,c\right),
\]
so that the zeroes of $v_0(z)$ divided by $1-c$ coincide with the zeroes of $F^{\text{alpha}}(z)$.

From Lemma \ref{lem:contiguous}, we deduce the recurrence relations
\begin{equation}\label{eq:JacobiEigenvalue}
w_{i-1}v_{i-1}(z) + (\la_i-z)v_i(z) + w_iv_{i+1}(z) = 0,
\end{equation}
valid for all $i\in\Z$, where
\begin{equation}\label{eq:JacobiParameters2}
w_i := \sqrt{c(\gamma+i)(\eta+i)}, \quad \la_i := -i(1+c) - c(\gamma+\eta)+1,
\end{equation}
though the integer $i$ has to be large enough, so that the product $(\gamma+i)(\eta+i)$ that goes inside the square root of equation \eqref{eq:JacobiParameters2} is nonnegative; this is true, in particular, whenever $i\ge 0$.
It follows from the relations \eqref{eq:JacobiEigenvalue} that if $z\in\C$ is a zero of $v_0(z)$, then $z$ is an eigenvalue of the Jacobi operator $J_{\la,w}$, with parameters~\eqref{eq:JacobiParameters2}, and with eigenvector $(v_i(z))_{i\ge 1}$.
We need to verify that this vector belongs to $\ell^2(\Z_{\ge 1})$, but this follows from the estimate
\begin{equation}\label{eq:stirling}
v_i(z) = O\left(c^{\frac{i}{2}}i^{\frac{\gamma+\eta}{2}-\frac{z}{1-c}}\right), \quad \text{ as } i\to\infty,
\end{equation}
which shows that $v_i(z)$ is exponentially small in $i$. The estimate itself is a result of the Stirling's formula and a well-known asymptotic expansion of the hypergeometric function ${}_2F_1(\alpha,\beta+x;\delta+x;z)$, as $x\to\infty$ (see e.g.~\cite[Eqn.~(15)]{Temme2003}).

The relations \eqref{eq:JacobiEigenvalue} equally show that if $z\in\C$ is a zero of $v_0(z)$, then $z/(1-c)$ is an eigenvalue of $\frac{1}{1-c}J_{\la,w}=J^{(\text{alpha})}$ with eigenvector $(v_i(z))_{i\ge 1}$. Hence, the zeroes of $F^{(\text{alpha})}(z)$ are eigenvalues of $J^{(\text{alpha})}$.

On the one hand, we want to prove that if $z$ is an eigenvalue of $J_{\la, w}$ with parameters~\eqref{eq:JacobiParameters2} (equivalently, if $z/(1-c)$ is an eigenvalue of $J^{(\text{alpha})}$), then $z$ a zero of $v_0(z)$.
Note that equations \eqref{eq:JacobiEigenvalue} hold true, for any $i\ge 2$, regardless of whether $z$ is a root of $v_0(z)$ or not.
Let us first assume that the eigenvector $(v_i'(z))_{i\ge 1}$ is not a multiple of $(v_i(z))_{i\ge 1}$.
Then, in particular, equations \eqref{eq:JacobiEigenvalue} should also be true with $v_i'(z)$ instead of $v_i(z)$, for all $i\ge 2$.
The Wrońskian of these two independent solutions is then a nonzero constant:
\[
w_i(v_i(z)v'_{i+1}(z) - v_{i+1}(z)v'_{i}(z)) = \text{const} \neq 0.
\]
But from the equation \eqref{eq:JacobiParameters2} and the estimate \eqref{eq:stirling}, it follows that both $w_i v_i(z)$ and $w_i v_{i+1}(z)$ tend to zero, as $i$ tends to infinity.
As a result, the sequence $(v'_i(z))_{i\ge 1}$ has to be unbounded, and hence, does not belong to $\ell^2(\Z_{\ge 1})$.
We conclude that, for any eigenvalue $z$ of $J_{\la,w}$, the corresponding eigenvector must be $(v_i(z))_{i\ge 1}$.
By looking at the first term of the eigenrelation $J_{\la,w}(v_i(z))_{i\ge 1} = (zv_i(z))_{i\ge 1}$, we have that $\la_1v_1(z)+w_1v_2(z)=zv_1(z)$.
Together with the equation \eqref{eq:JacobiEigenvalue} for $i=1$, we deduce that $w_0v_0(z)=0$, and therefore $v_0(z)=0$, as desired.
This ends the proof in the alpha case.
\end{proof}

The following technical lemma is needed later; we give a sketch of proof in the appendix.

\begin{lemma}\label{lem:finite_order}
Let $\gamma,\eta > 0$, $c\in (0,1)$.
The entire functions $F^{\text{alpha}}(z)$ and $F^{\text{planch}}(z)$, depending on the previous parameters and defined in Theorem \ref{thm:char_fns}, are of finite order.
\end{lemma}

\section{The inverse of the formal Laplace transform}
\label{sec:analytic}

The formal Laplace transform $\widetilde{\Lap}$ admits the unique inverse map $\widetilde{\Lap}^{-1}$ that takes as input a power series on $z^{-1}$ without constant coefficient and acts as follows:
\begin{equation}\label{eq:definition_laplace}
\widetilde{\Lap}^{-1}\bigg\{ \sum_{n=0}^\infty{a_nz^{-n-1}} \bigg\} = \sum_{n=0}^\infty{ \frac{a_n}{n!}t^n }.
\end{equation}

\begin{lemma}\label{lem:lap_derivative}
For any power series $f(z)$ on $z^{-1}$ without constant coefficient, we have
\begin{equation*}
\widetilde{\Lap}^{-1}\Big\{ f(z) \Big\} = -\frac{1}{t}\cdot\widetilde{\Lap}^{-1}\Big\{ f'(z) \Big\}.
\end{equation*}
\end{lemma}
\begin{proof}
Immediately follows from \cref{eq:definition_laplace}.
\end{proof}

\begin{theorem}\label{thm:inverse_laplace}
Assume that $F(z) = \frac{1}{\Gamma(z)}\cdot G(z)$ is an entire function of finite order whose nonzero roots can be labeled $\ell_1>\ell_2>\dots$, and satisfy $\ell_{k+1}=-k+O\left(k^{-N}\right)$, for any $N\ge 1$, as $k\to\infty$.
Moreover, assume that $G(z)$ admits an asymptotic expansion $G(z)\sim 1+\sum_{n=0}^\infty{a_n z^{-n-1}}$, as $z\to\infty$ along the $x$-axis. Then
\[
\widetilde{\Lap}^{-1}\big\{\ln{G(z)}\big\} = \sum_{n=0}^\infty \frac{t^n}{(n+1)!} \sum_{k=0}^\infty \Big( (-k)^{n+1} - (\ell_{k+1})^{n+1} \Big).
\]
The left-hand side of the above equality is a notation that really stands for the inverse of the formal Laplace transform~\eqref{eq:definition_laplace} applied to $\ln(1+x)=\sum_{n=1}^\infty{(-1)^{n-1}\frac{x^n}{n}}$, where $x=\sum_{n=0}^\infty{a_n z^{-n-1}}$.
\end{theorem}
\begin{proof}
Since $\ell_{k+1}=-k+O\left(k^{-1}\right)$, as $k\to\infty$, then
$\sum_{k=0}^\infty{\frac{1}{|\ell_{k+1}|}}=\infty$, but $\sum_{k=0}^\infty{\frac{1}{|\ell_{k+1}|^2}}<\infty$.
Therefore, by Hadamard's factorization theorem,
\[
F(z) = z^m e^{Q(z)} \prod_{k=0}^\infty\left( 1 - \frac{z}{\ell_{k+1}} \right)e^{z/\ell_{k+1}},
\]
where the nonnegative integer $m$ is the multiplicity of $z=0$ as a zero of $F(z)$, and $Q(z)$ is a polynomial.
After taking the logarithmic derivative, we obtain
\begin{equation}\label{eq:expansion_1_planch}
\Big( \ln{F(z)} \Big)' = \frac{m}{z} + Q'(z) + \sum_{k=0}^\infty \left\{ \frac{1}{z-\ell_{k+1}} + \frac{1}{\ell_{k+1}} \right\}.
\end{equation}

On the other hand, the expansion of the digamma function $\left( \ln{\Gamma(z)} \right)'=\Gamma'(z)/\Gamma(z)$ is
\begin{equation}\label{eq:expansion_2_planch}
\left( \ln{\Gamma(z)} \right)' = -\gamma + \sum_{k=0}^\infty\left\{ \frac{1}{k+1} - \frac{1}{z+k} \right\},
\end{equation}
where $\gamma>0$ is the Euler-Mascheroni constant; this expansion is valid for all $z$ outside the set of nonpositive integers.

By adding \cref{eq:expansion_1_planch,eq:expansion_2_planch}, we get
\begin{equation}\label{eq:expansion_3_planch}
\Big( \ln{F(z)} \Big)' + \left( \ln{\Gamma(z)} \right)' = \frac{m}{z} + Q'(z) - \gamma + \sum_{k=0}^\infty\left\{  \frac{1}{\ell_{k+1}} + \frac{1}{k+1} + \frac{1}{z-\ell_{k+1}} - \frac{1}{z+k} \right\}.
\end{equation}
Note that $\sum_{k=0}^\infty\left(\frac{1}{\ell_{k+1}} + \frac{1}{k+1}\right)$ converges; therefore, the sum on the right-hand side of \eqref{eq:expansion_3_planch} can be broken down into two sums.
On the other hand, since $G(z)=F(z)\cdot\Gamma(z)$, the left-hand side of \eqref{eq:expansion_3_planch} is equal to $(\ln{G(z)})'$. The equality will be true for values of $z$ that are not the roots $\ell_{k+1}$ of $F(z)$ and are not in the set of nonpositive integers. Thus,
\begin{equation}\label{eq:expansion_4_planch}
\Big( \ln{G(z)} \Big)' = \frac{m}{z} + R(z) + \sum_{k=0}^\infty\left\{  \frac{1}{z-\ell_{k+1}} - \frac{1}{z+k} \right\},
\end{equation}
where $R(z)$ is the polynomial
\[
R(z) := Q'(z) - \gamma + \sum_{k=0}^\infty\left\{  \frac{1}{\ell_{k+1}} + \frac{1}{k+1} \right\}.
\]
By the assumption of $G(z)$, we have $\lim_{z\to+\infty}{z\cdot\big(\ln{G(z)}\big)'}=0$.
We claim that the same is true of $z$ times the summation in the RHS of \cref{eq:expansion_4_planch}.
Indeed, if $z>0$ and $k$ is large enough so that $\ell_{k+1}<0$ and $|\ell_{k+1}+k|<\text{const}_1/k$, for some positive constant, then
\[
\left| z\cdot\left(\frac{1}{z-\ell_{k+1}} - \frac{1}{z+k} \right) \right| = \frac{|z|\cdot|\ell_{k+1}+k|}{|z-\ell_{k+1}|\cdot|z+k|} \le \frac{z\cdot\text{const}_1/k}{z\cdot (z+k)}\le\frac{\text{const}_1}{k\cdot(2\sqrt{zk})} = z^{-\frac{1}{2}}\cdot\frac{\text{const}_2}{k^{3/2}}.
\]
It follows that $m+zR(z)$ converges to zero, as $z\to\infty$, and therefore it is zero.
As a result,
\begin{equation}\label{eq:expansion_5_planch}
\Big( \ln{G(z)} \Big)' = \sum_{k=0}^\infty\left\{  \frac{1}{z-\ell_{k+1}} - \frac{1}{z+k} \right\}.
\end{equation}
Since $\ell_{k+1}=-k+O\left(k^{-N}\right)$, for all $N\ge 1$, then the sums $\sum_{k=0}^\infty\Big( \big(\ell_{k+1}\big)^n - \big(-k\big)^n \Big)$ converge, for all $n\ge 1$. By expanding the RHS of \cref{eq:expansion_5_planch} around $z=\infty$, we obtain
\begin{equation}\label{eq:expansion_6_planch}
\Big( \ln{G(z)} \Big)' = \sum_{n=1}^\infty z^{-n-1} \sum_{k=0}^\infty\Big( \big(\ell_{k+1}\big)^n - \big(-k\big)^n \Big).
\end{equation}
Finally, by \cref{lem:lap_derivative} and \cref{eq:expansion_6_planch}, we deduce that
\[
\widetilde{\Lap}^{-1}\Big\{\ln{G(z)}\Big\}
= -\frac{1}{t}\cdot\widetilde{\Lap}^{-1}\bigg\{ \Big( \ln{G(z)} \Big)' \bigg\}
= \sum_{n=0}^\infty \frac{t^n}{(n+1)!} \sum_{k=0}^\infty \left( (-k)^{n+1} - \big(\ell_{k+1}\big)^{n+1} \right),
\]
as desired.
\end{proof}

\begin{corollary}\label{cor:formal_inverse_transform}
Let $\gamma,\eta>0$ be arbitrary.
Let $J^{(\text{planch})}$ be the Jacobi matrix that depends on these parameters, defined by \cref{eq:Jacobiplanch}, and let $\ell_1^{(\text{planch})}>\ell_2^{(\text{planch})}>\dots$ be its eigenvalues.
Then
\[
\widetilde{\Lap}^{-1}\bigg\{ \ln\Big({}_1F_1\left(\gamma;z;-\eta\right)\Big) \bigg\} = \sum_{n=0}^\infty \frac{t^n}{(n+1)!} \sum_{k=0}^\infty \left( (-k)^{n+1} - \big(\ell^{(\text{planch})}_{k+1}\big)^{n+1} \right).
\]
\end{corollary}
\begin{proof}
We will apply \cref{thm:inverse_laplace} to $G(z)={}_1F_1\left(\gamma;z;-\eta\right)$.
We need to check the conditions of that theorem.
The fact that $F(z)=\frac{1}{\Gamma(z)}\cdot G(z)$ is entire and of finite order is proved in \cref{lem:finite_order}, while the fact that the zeroes of $F(z)$ are simple and satisfy the desired asymptotic relation follows from \cref{lem:roots_lemma,thm:char_fns}.
Finally, it is easy to verify from the power series definition of ${}_1F_1(\gamma;z;-\eta)$ that this function has a complete asymptotic expansion, as $z\to\infty$, beginning with ${}_1F_1(\gamma;z;-\eta)=1+O(z^{-1})$, as needed.
\end{proof}

\begin{theorem}\label{thm:inverse_laplace_2}
Assume that $F(z) = \frac{1}{\Gamma(z)}\cdot G(z)$ is entire, of finite order, whose nonzero roots $\ell_1>\ell_2>\dots$ satisfy $\ell_k=-k+O\left(k^{-N}\right)$, for any $N\ge 1$, as $k\to\infty$.
Assume also that $G(z)$ admits an asymptotic expansion $G(z)\sim 1+\sum_{n=0}^\infty{a_n z^{-n-1}}$, as $z\to\infty$ along the $x$-axis. Then
\[
\Lapformal^{-1}\big\{\ln{G(z)}\big\} = \sum_{n=0}^\infty \frac{t^n}{(n+1)!} \sum_{k=1}^\infty{\Big( (-k)^{n+1} - (\ell_k)^{n+1} \Big)}.
\]
\end{theorem}
\begin{proof}
The statement of the theorem is quite similar to that of \cref{thm:inverse_laplace}. In fact, the only difference in the assumptions of the theorem is that now we have the asymptotic relation $\ell_k=-k+O\big(k^{-N}\big)$ instead of $\ell_{k+1}=-k+O\big(k^{-N}\big)$.
The proof of the present theorem is similar; in fact, we will only point out the places in the proof where we need modifications.

As in the proof of \cref{thm:inverse_laplace}, we find
\[
\Big( \ln{F(z)} \Big)' = \frac{m}{z} + Q'(z) + \sum_{k=1}^\infty \left\{ \frac{1}{z-\ell_k} + \frac{1}{\ell_k} \right\},
\]
for some $m\in\Z_{\ge 0}$ and polynomial $Q(z)$.
Next, write the expansion of $\big(\ln{\Gamma(z)}\big)'$ as
\[
\big(\ln{\Gamma(z)}\big)' = -\gamma-\frac{1}{z} + \sum_{k=1}^\infty\left\{ \frac{1}{k} - \frac{1}{z+k} \right\}.
\]
By adding these equalities, we deduce
\[
\Big( \ln{G(z)} \Big)' = \frac{m-1}{z} + R(z) + \sum_{k=1}^\infty \left\{ \frac{1}{z-\ell_k} - \frac{1}{z+k} \right\},
\]
where $R(z):=Q'(z)-\gamma+\sum_{k=1}^\infty{\big\{\frac{1}{\ell_k}+\frac{1}{k}\big\}}$.
Using the asymptotic relation $\ell_k=-k+O\left(k^{-N}\right)$, for any $N\ge 1$, as $k\to\infty$, as in the proof of \cref{thm:inverse_laplace}, we verify $\lim_{z\to\infty}{(m-1+zR(z))}=0$.
This limit implies that the polynomial $m-1+zR(z)$ is identically zero, therefore
\[
\Big( \ln{G(z)} \Big)' = \sum_{k=1}^\infty \left\{ \frac{1}{z-\ell_k} - \frac{1}{z+k} \right\}.
\]
The proof then proceeds along the same lines as the proof of \cref{thm:inverse_laplace}.
\end{proof}

\begin{corollary}\label{cor:formal_inverse_transform_2}
Let $\gamma,\eta>0$, $c\in(0,1)$ be arbitrary.
Let $J^{(\text{alpha})}$ be the Jacobi matrix that depends on these parameters, defined by \cref{eq:Jacobialpha}, and let $\ell_1^{(\text{alpha})}>\ell_2^{(\text{alpha})}>\dots$ be the eigenvalues of $J^{(\text{alpha})}$.
If $\ell_k^{(\text{alpha})}=-k+O(k^{-N})$, for any $N>0$, as $k\to\infty$, then
\[
\begin{aligned}
\Lapformal^{-1}\bigg\{ \ln\Big( (1-c)^\gamma\cdot{}_2F_1\left(\gamma,z-\eta;z;c\right) \Big) \bigg\} &= \Lapformal^{-1}\bigg\{ \ln\Big( (1-c)^\eta\cdot{}_2F_1\left(\eta,z-\gamma;z;c\right) \Big) \bigg\}\\
&= \sum_{n=0}^\infty \frac{t^n}{(n+1)!} \sum_{k=1}^\infty \Big( (-k)^{n+1} - \big(\ell^{(\text{alpha})}_k\big)^{n+1} \Big).
\end{aligned}
\]
In particular, this last equality holds when $\gamma=\eta$.
\end{corollary}
\begin{proof}
The first equality is Pfaff's transformation, namely \cref{lem:pfaff}.
The second equality results from \cref{thm:inverse_laplace_2} applied to $G(z)=(1-c)^\gamma\cdot{}_2F_1\left(\gamma,z-\eta;z;c\right)$, so it only remains to verify the assumptions.
The function $F(z)=\frac{1}{\Gamma(z)}\cdot G(z)$ is entire and of finite order, due to \cref{lem:finite_order}. The zeroes of $F(z)$ are simple and satisfy the desired asymptotic relations because of \cref{lem:roots_lemma}, \cref{thm:char_fns}, and the assumption of the corollary. Finally, ${}_2F_1\left(\gamma,z-\eta;z;c\right)$ is known to have an asymptotic expansion, as $z\to\infty$, that begins with ${}_2F_1\left(\gamma,z-\eta;z;c\right)\sim (1-c)^{-\gamma}\cdot\big(1+O(z^{-1})\big)$, see e.g.~\cite{Temme2003}, already cited in the proof of \cref{thm:char_fns}.
Therefore, $(1-c)^\gamma\cdot{}_2F_1\left(\gamma,z-\eta;z;c\right)\sim 1+O(z^{-1})$, as required.
\end{proof}

\begin{theorem}\label{thm:inverse_laplace_beta}
Let $M\in\Z_{\ge 1}$, $c\in(0,1)$ be arbitrary.
Let $J^{(\text{beta})}$ be the symmetric tridiagonal $M\times M$ matrix that depends on these parameters and is defined by \cref{eq:Jacobibeta}.
The eigenvalues of $J^{(\text{beta})}$ or, equivalently, the roots of the monic degree $M$ polynomial $F^{\text{beta}}(z)$ are real; let us denote them $\ell_1^{(\text{beta})},\dots,\ell_M^{(\text{beta})}$.
Then
\begin{equation}\label{eq:inverse_fourier_beta}
\Lapformal^{-1}\big\{\ln{{}_2F_1\left(-M,\gamma; z; c\right)}\big\}
= \sum_{n=0}^\infty \frac{t^n}{(n+1)!}
\left\{ \sum_{k=1}^M{(1-k)^{n+1}} - 
\sum_{k=1}^M{\big(\ell_k^{(\text{beta})}\big)^{n+1}} \right\}.
\end{equation}
\end{theorem}

\begin{proof}
Since $F^{\text{beta}}(z)$ is a monic polynomial with roots $\ell_1^{(\text{beta})},\dots,\ell_M^{(\text{beta})}$, we have
\[
F^{\text{beta}}(z) = \frac{\Gamma(z+M)}{\Gamma(z)}{}_2F_1\left(-M,\gamma; z; c\right) = \prod_{k=1}^M{\big(z - \ell_k^{(\text{beta})}\big)}.
\]
Since $\frac{\Gamma(z+M)}{\Gamma(z)} = \prod_{k=1}^M\big(z-(1-k)\big)$, we deduce
\[
\Big( \ln{{}_2F_1\left(-M,\gamma; z; c\right)} \Big)' = \sum_{k=1}^M{\left\{ \frac{1}{z-\ell_k^{(\text{beta})}} - \frac{1}{z-(1-k)} \right\}}.
\]
Finally, \cref{lem:lap_derivative} and the previous equation lead to
\begin{align*}
\Lapformal^{-1}\Big\{ \ln{{}_2F_1\left(-M,\gamma; z; c\right)} \Big\}
&= -\frac{1}{t}\cdot\widetilde{\Lap}^{-1}\bigg\{ \Big( \ln{{}_2F_1\left(-M,\gamma; z; c\right)} \Big)' \bigg\}\\
&= \sum_{n=0}^\infty \frac{t^n}{(n+1)!} \sum_{k=1}^M{ \left( (1-k)^{n+1} - \big(\ell_k^{(\text{beta})}\big)^{n+1} \right) },
\end{align*}
which is the desired \cref{eq:inverse_fourier_beta}.
\end{proof}

\section{Densities of the limits of pure Jack measures}
\label{sec:densities}

In this section, as the main application of our main \cref{thm:r_transform}, we find densities for the limits of empirical measures of pure Jack measures, i.e., for $\mu_{\Alpha}^{\gamma;c;\eta}$, $\mu_{\Beta}^{\gamma;c;M}$ and $\mu_{\Planch}^{\gamma;\eta}$, namely the measures from \cref{thm:application1}, \cref{thm:application2} and \cref{thm:application3}, respectively.

\begin{theorem}[Limit of pure Plancherel Jack measures]\label{thm:application_plancherel}
Let $\gamma,\eta>0$ be arbitrary.
Let $\mu_{\Planch}^{\gamma;\eta}$ be the probability measure from \cref{thm:application3}.
Let $\ell_1^{(\text{planch})}>\ell_2^{(\text{planch})}>\cdots$ be the eigenvalues of $J^{(\text{planch})}$ or, equivalently, the zeroes of $F^{\text{planch}}(z)$\footnote{The parameters $\gamma,\eta$ that go into the definitions of the Jacobi operator $J^{(\text{planch})}$ and the meromorphic function $F^{\text{planch}}(z)$ are the same as the ones in $\mu_{\Planch}^{\gamma;\eta}$.}; see \cref{eq:Jacobiplanch,thm:char_fns} for the definitions.
Then $\gamma\ge\ell_1^{(\text{planch})}$ and $\ell_k^{(\text{planch})}\ge 1+\ell_{k+1}^{(\text{planch})}$, for all $k\ge 1$.
Moreover, $\mu_{\Planch}^{\gamma;\eta}$ is absolutely continuous and has density
\begin{equation}\label{eq:density_planch}
\frac{\dd\mu_{\Planch}^{\gamma;\eta}}{\dd x} = f_{\Planch}^{\gamma;\eta}(x) := \frac{1}{\gamma}\sum_{k=0}^\infty{ \mathbf{1}_{\big[1-\ell_k^{(\text{planch})},\, -\ell_{k+1}^{(\text{planch})}\big]}(x) },\quad x\in\R,
\end{equation}
where we set $\ell_0^{(\text{planch})}:=\gamma+1$.
\end{theorem}
\begin{proof}
By \cref{thm:application3}, the measure $\mu_{\Planch}^{\gamma;\eta}$ has finite moments of all orders and, consequently, all its quantized $\gamma$-cumulants are well-defined; moreover, they equal $\kappa_n=\eta\cdot\delta_{n,1}$.
Then \cref{eq:1} gives
\[
1+\sum_{n\ge 1}{\frac{c_n}{\gamma^{\uparrow n}}z^n} = \exp(\eta z) = \sum_{n\ge 0}{\frac{\eta^n}{n!}z^n},
\]
therefore
\[
c_n = \frac{\gamma^{\uparrow n}}{n!}\eta^n,\ \text{ for all }n\ge 1.
\]
Next, the LHS of \cref{eq:2} is
\[
1+\sum_{n=1}^M{ \frac{(-1)^n\frac{\gamma^{\uparrow n}}{n!}\eta^n}{z^{\uparrow n}} } = \sum_{n=0}^M{ \frac{\gamma^{\uparrow n}}{z^{\uparrow n}}\cdot\frac{(-\eta)^n}{n!} } = {}_1F_1(\gamma;z;-\eta).
\]
Denote by $M(s) := 1+\sum_{n=1}^\infty{\int_{\R}{x^n\mu_{\Planch}^{\gamma;\eta}(\dd x)}\cdot s^n}$ the formal moment generating function of $\mu_{\Planch}^{\gamma;\eta}$.
By \cref{thm:r_transform} and \cref{cor:formal_inverse_transform}, we deduce the following equality of formal power series:
\[
\frac{1}{\gamma t}\sum_{k=0}^\infty{ \Big( e^{-kt} - e^{\ell_{k+1}^{(\text{planch})}t} \Big) }\cdot (1-e^{-t}) = M(-t)-\frac{1}{\gamma t}\left( e^{\gamma t}-1 \right).
\]
The simple manipulation
\[
\sum_{k=0}^\infty{ \Big( e^{-kt} \Big)\cdot (1-e^{-t}) } = \sum_{k=0}^\infty{ \Big( e^{-kt} - e^{-(k+1)t} \Big) } = 1
\]
leads to:
\[
M(t) = \frac{e^{-\ell_1^{(\text{planch})} t} - e^{-\gamma t}}{\gamma t} + \sum_{k=1}^\infty{ \frac{e^{-\ell_{k+1}^{(\text{planch})}t} - e^{\big(1-\ell_k^{(\text{planch})}\big)t}}{\gamma t} }.
\]
This formal equality of power series, along with the fact that all moments of $\mu_{\Planch}^{\gamma;\eta}$ are finite, implies that the characteristic function of $\mu_{\Planch}^{\gamma;\eta}$ equals
\begin{equation}\label{eq:char_function_planch}
\int_{\R}{e^{itx}\mu_{\Planch}^{\gamma;\eta}(\text{d}x)} = \frac{1}{\gamma}\cdot\left\{ \frac{e^{-i\ell_1^{(\text{planch})} t} - e^{-i\gamma t}}{it} + \sum_{k=1}^\infty{ \frac{e^{-i\ell_{k+1}^{(\text{planch})}t} - e^{i\big(1-\ell_k^{(\text{planch})}\big)t}}{it} } \right\}.
\end{equation}

Let us assume for now that
\begin{equation}\label{eq:ineq_plancherel}
\gamma\ge\ell_1^{(\text{planch})}\ \text{ and }\ \ell_k^{(\text{planch})} \ge 1+\ell_{k+1}^{(\text{planch})}, \text{ for all }k\ge 1.
\end{equation}
Then $1-\ell_k^{(\text{planch})} \le -\ell_{k+1}^{(\text{planch})}$, for all $k\ge 0$, so that all intervals $\big[1-\ell_k^{(\text{planch})},\, -\ell_{k+1}^{(\text{planch})}\big]$, for $k\ge 0$, are well-defined and disjoint (if $1-\ell_k^{(\text{planch})} = -\ell_{k+1}^{(\text{planch})}$, for some $k$, then we can remove the corresponding indicator function from the density).
Moreover, the function defined by \cref{eq:density_planch} is the density of a probability measure, since it integrates to:
\begin{align*}
\frac{1}{\gamma}\cdot\sum_{k=0}^\infty{\left( \big( -\ell_{k+1}^{(\text{planch})} \big) - \big(1 - \ell_k^{(\text{planch})}\big) \right)} &= \frac{1}{\gamma}\cdot\sum_{k=0}^\infty{\left( \ell_k^{(\text{planch})} -\ell_{k+1}^{(\text{planch})} - 1 \right)}\\
&= \frac{1}{\gamma}\cdot\lim_{K\to\infty} \sum_{k=0}^K{\left( \ell_k^{(\text{planch})} -\ell_{k+1}^{(\text{planch})} - 1 \right)}\\
&= \frac{1}{\gamma}\cdot\lim_{K\to\infty} \left( \ell_0^{(\text{planch})} - \ell_{K+1}^{(\text{planch})}-(K+1) \right)\\
&= \frac{1}{\gamma}\cdot\big(\ell_0^{(\text{planch})} - 1\big) = 1,
\end{align*}
where for the fourth equality, we used $\lim_{K\to\infty}{\big(\ell_{K+1}^{(\text{planch})}+K\big)} = 0$, as given by \cref{lem:roots_lemma}.
Next, note that $\int_{\R}{e^{itx}f_{\Planch}^{\gamma;\eta}(x)dx}$ is equal to the RHS of \cref{eq:char_function_planch}.
Since the characteristic function uniquely determines the underlying probability measure, it would follow that $\mu_{\Planch}^{\gamma;\eta}$ is absolutely continuous with density $f_{\Planch}^{\gamma;\eta}(x)$, as desired.

It remains to prove \cref{eq:ineq_plancherel}. If we set
\begin{multline*}
g_{\Planch}^{\gamma;\eta}(x) := \frac{1}{2\gamma}\bigg\{ \left( \sgn(x+\gamma) - \sgn\big(x+\ell_1^{(\text{planch})}\big) \right) \\
+ \sum_{k=1}^\infty{ \left( \sgn\big(x+\ell_k^{(\text{planch})}-1\big) - \sgn\big( x+\ell_{k+1}^{(\text{planch})} \big) \right) } \bigg\},
\end{multline*}
then $\int_\R{e^{itx}g_{\Planch}^{\gamma;\eta}(x)\dd x}$ equals the RHS of \cref{eq:char_function_planch}.
This means that $g_{\Planch}^{\gamma;\eta}(x)$ is a density of the probability measure $\mu_{\Planch}^{\gamma;\eta}$.
For the sake of contradiction, assume $-\ell_1^{(\text{planch})} < -\gamma$. Since $-\ell_1^{(\text{planch})}<1-\ell_k^{(\text{planch})}$ and $-\ell_1^{(\text{planch})}<-\ell_{k+1}^{(\text{planch})}$, for all $k\ge 1$, it is easy to check that $g_{\Planch}^{\gamma;\eta}\big(-\ell_1^{(\text{planch})}+\epsilon\big)<0$, for sufficiently small $\epsilon>0$. But this is impossible, therefore $-\ell_1^{(\text{planch})}\ge -\gamma$, which proves the first equality in \eqref{eq:ineq_plancherel}.
Similarly, $-\ell_{k+1}^{(\text{planch})}\ge 1-\ell_k^{(\text{planch})}$, for all $k\ge 1$, by induction on $k$.
This ends the proof of \eqref{eq:ineq_plancherel} and the theorem.
\end{proof}

\begin{theorem}[Limit of pure alpha Jack measures]\label{thm:application_alpha}
Let $\gamma,\eta>0$ and $c\in (0,1)$ be arbitrary.
Let $\mu_{\Alpha}^{\gamma;c;\eta}$ be the probability measure from \cref{thm:application1}.
Let $\ell_1^{(\text{alpha})}>\ell_2^{(\text{alpha})}>\cdots$ be the eigenvalues of $J^{(\text{alpha})}$ or, equivalently, the zeroes of $F^{\text{alpha}}(z)$.
Moreover, assume that $\ell_{k}^{(\text{alpha})}=-k+O(k^{-N})$, for any $N\ge 1$, as $k\to\infty$.\footnote{This assumption holds when $\gamma=\eta$ and we believe that it holds for all $\gamma,\eta>0$.}
Then $-1\ge\ell_1^{(\text{alpha})}$ and $\ell_k^{(\text{alpha})}\ge 1+\ell_{k+1}^{(\text{alpha})}$, for all $k\ge 1$.
Moreover, $\mu_{\Alpha}^{\gamma;c;\eta}$ is absolutely continuous and has density
\[
\frac{\dd\mu_{\Alpha}^{\gamma;c;\eta}}{\dd x} = \frac{1}{\gamma}\sum_{k=-1}^\infty{ \mathbf{1}_{\big[1-\ell_k^{(\text{alpha})},\, -\ell_{k+1}^{(\text{alpha})}\big]}(x) },\quad x\in\R,
\]
if we set $\ell_0^{(\text{alpha})}:=0$ and $\ell_{-1}^{(\text{alpha})}:=\gamma+1$, for convenience.
\end{theorem}
\begin{proof}
By \cref{thm:application3}, the quantized $\gamma$-cumulants of $\mu_{\Alpha}^{\gamma;c;\eta}$ are $\kappa_n=\eta\left(\frac{c}{1-c}\right)^n$, thus \cref{eq:1} gives
\begin{align*}
1+\sum_{n\ge 1}{\frac{c_n}{\gamma^{\uparrow n}}z^n}
= \exp\left( \eta\sum_{n\ge 1}{\frac{1}{n} \left( \frac{cz}{1-c} \right)^n} \right)
&= \exp\bigg(-\eta\cdot\ln\Big(1-\frac{cz}{1-c}\Big)\bigg)\\
&= \Big( 1-\frac{cz}{1-c} \Big)^{-\eta}
= \sum_{n\ge 0}{\frac{\eta^{\uparrow n}}{n!}\Big(\frac{c}{1-c}\Big)^n z^n},
\end{align*}
where the last equality follows from the binomial theorem.
As a result,
\[
c_n = \frac{\gamma^{\uparrow n}\eta^{\uparrow n}}{n!}\Big(\frac{c}{1-c}\Big)^n,\ \text{ for all }n\ge 1.
\]
Then the LHS of \cref{eq:2} is
\begin{multline*}
1+\sum_{n\ge 1}{ \frac{(-1)^n\,\frac{\gamma^{\uparrow n}\eta^{\uparrow n}}{n!}\left(\frac{c}{1-c}\right)^n}{z^{\uparrow n}} }
= \sum_{n\ge 0}{ \frac{\gamma^{\uparrow n}\eta^{\uparrow n}}{z^{\uparrow n}}\cdot\frac{\left(\frac{c}{c-1}\right)^n}{n!} }
= {}_2F_1\Big(\gamma,\eta;z;\frac{c}{c-1}\Big)\\
= (1-c)^\gamma\cdot{}_2F_1\left(\gamma,z-\eta;z;c\right)
= (1-c)^\eta\cdot{}_2F_1\left(z-\gamma,\eta;z;c\right),
\end{multline*}
where we have used \cref{lem:pfaff}.
The advantage of the last two formulas is that the corresponding series define analytic functions, for any $c\in(0,1)$.
Next, if we let $M(s)$ be the formal moment generating function of $\mu_{\Alpha}^{\gamma;c;\eta}$, then \cref{thm:r_transform,cor:formal_inverse_transform_2} imply the following equality of formal power series:
\[
\frac{1}{\gamma t}\sum_{k=1}^\infty{ \Big( e^{-kt} - e^{\ell_k^{(\text{alpha})}t} \Big) }\cdot (1-e^{-t}) = M(-t)-\frac{1}{\gamma t}\left( e^{\gamma t}-1 \right).
\]
The remaining of the proof then follows the proof of \cref{thm:application_plancherel}, starting by computing that the characteristic function of $\mu_{\Alpha}^{\gamma;c;\eta}$ is
\[
\int_{\R}{e^{itx}\mu_{\Alpha}^{\gamma;c;\eta}(\dd x)} = \frac{1}{\gamma}\cdot\left\{ \frac{e^{-i\ell_1^{(\text{alpha})} t} - e^{it} + 1 - e^{-i\gamma t}}{it} + \sum_{k=1}^\infty{ \frac{e^{-i\ell_{k+1}^{(\text{alpha})}t} - e^{i\big(1-\ell_k^{(\text{alpha})}\big)t}}{it} } \right\},
\]
and from there, the formula for the density follows.
\end{proof}

\begin{theorem}[Limit of pure beta Jack measures]\label{thm:application_beta}
Let $M\in\Z_{\ge 1}$, $c\in(0,1)$.
Let $\mu_{\Beta}^{\gamma;c;M}$ be the probability measure from \cref{thm:application2}.
Let $\ell_1^{(\text{beta})}\ge\dots\ge\ell_M^{(\text{beta})}$ be the eigenvalues of $J^{(\text{beta})}$ or, equivalently, the roots of $F^{\text{beta}}(z)$.
Then $\gamma\ge\ell_1^{(\text{beta})}$ and $\ell_k^{(\text{beta})}\ge 1+\ell_{k+1}^{(\text{beta})}$, for all $1\le k\le M-1$; in particular, the eigenvalues are pairwise distinct.
Moreover, $\mu_{\Beta}^{\gamma;c;M}$ has density
\begin{equation*}
\frac{\dd\mu_{\Beta}^{\gamma;c;M}}{\dd x} = \frac{1}{\gamma}\sum_{k=0}^{M-1}{ \mathbf{1}_{\big[1-\ell_k^{(\text{beta})},\, -\ell_{k+1}^{(\text{beta})}\big]}(x) },\quad x\in\R,
\end{equation*}
if we set $\ell_0^{(\text{beta})}:=\gamma+1$.
\end{theorem}
\begin{proof}
By \cref{thm:application2}, the quantized $\gamma$-cumulants of $\mu_{\Beta}^{\gamma;c;M}$ are $\kappa_n=(-1)^{n-1}Mc^n$.
Then \cref{eq:1} gives
\[
1+\sum_{n\ge 1}{\frac{c_n}{\gamma^{\uparrow n}}z^n} = \exp\left( \sum_{n\ge 1}{\frac{(-1)^{n-1}Mc^nz^n}{n}} \right) = (1+cz)^M,
\]
and leads to
\[
c_n = \begin{cases}
    \dbinom{M}{n} c^n \gamma^{\uparrow n}, &\text{if }n=1,\dots,M,\\
    0, &\text{if }n>M.
\end{cases}
\]
Next, the LHS of \cref{eq:2} is
\[
1+\sum_{n=1}^M{ \frac{(-1)^n \dbinom{M}{n} c^n \gamma^{\uparrow n}}{z^{\uparrow n}} } = \sum_{n=0}^M{ \frac{(-M)^{\uparrow n} \gamma^{\uparrow n}}{z^{\uparrow n}}\cdot\frac{c^n}{n!} } = {}_2F_1(-M,\gamma;z;c).
\]
Then, if we let $\mathfrak{M}(s)$ be the formal moment generating function of $\mu_{\Beta}^{\gamma;c;M}$, and due to the already proved \cref{eq:inverse_fourier_beta}, we deduce from \cref{thm:r_transform} the following equality of power series:
\[
\frac{1}{\gamma t}\sum_{k=1}^M{ \Big( e^{(1-k)t} - e^{\ell_k^{(\text{beta})}t} \Big) }\cdot (1-e^{-t}) = \mathfrak{M}(-t)-\frac{1}{\gamma t}\left( e^{\gamma t}-1 \right).
\]
From here, the proof mimics the proof of \cref{thm:application_plancherel}. In fact, the previous equality gives
\begin{equation*}
\int_{\R}{e^{itx}\mu_{\Beta}^{\gamma;c;M}(\dd x)} = \frac{1}{\gamma}\cdot\left\{ \frac{e^{-i\ell_1^{(\text{beta})} t} - e^{-i\gamma t}}{it} + \sum_{k=1}^{M-1}{ \frac{e^{-i\ell_{k+1}^{(\text{beta})}t} - e^{i\big(1-\ell_k^{(\text{beta})}\big)t}}{it} } \right\},
\end{equation*}
as well as the density formula for $\mu_{\Beta}^{\gamma;c;M}$. Hence, the proof is finished.
\end{proof}

\section{Appendix: sketch of proof of Lemma~\ref{lem:finite_order}}

\subsection{The Plancherel case}

Let $\gamma,\eta>0$ be arbitrary.
We want to prove that
\begin{equation}\label{eq:plancherel_order}
F^{\text{planch}}(z) = \frac{1}{\Gamma(z)}\cdot{}_1F_1\left(\gamma;z;-\eta\right)
= \frac{1}{\Gamma(z)}\sum_{n\ge 0}{ \frac{\gamma^{\uparrow n}}{z^{\uparrow n}}\frac{(-\eta)^n}{n!} }
\end{equation}
has finite order, i.e. that there exist constants $A,r>0$ such that
\begin{equation}\label{eq:def_order}
\big| F^{\text{planch}}(z) \big| \le \exp\big( |z|^A \big),\quad\text{for all }|z|>r.
\end{equation}
It is enough to prove \eqref{eq:def_order} with an additional assumption that $|z|\notin\Z_{\ge 0}$.
The key step is to bound the general $n$-th term in the sum \eqref{eq:plancherel_order}, which we denote by $a_n$. Note that
\begin{equation}\label{eq:order_1}
|a_n| = \left| \frac{\gamma^{\uparrow n}}{z^{\uparrow n}}\frac{(-\eta)^n}{n!} \right| = \frac{\gamma^{\uparrow n}\eta^n}{n!}\cdot\frac{1}{\left| z^{\uparrow n} \right|} \le \frac{(\lceil \gamma \rceil \eta)^n}{\left| z^{\uparrow n} \right|},\quad\text{for all }n\in\Z_{\ge 0}.
\end{equation}
Then there exists a unique integer $m\ge 2$ such that
\[
m-1>|z|>m-2.
\]
Then, for $n\ge m$, we can lower bound $\left| z^{\uparrow n} \right|$ as follows:
\begin{align}
\left| z^{\uparrow n} \right| &= |z|\cdot |z+1|\cdots |z+n-1|\nonumber\\
&\ge |z|\cdots(|z|-m+3)\cdot|z+m-2|\cdot|z+m-1|\cdot (m-|z|)\cdots (n-1-|z|)\nonumber\\
&\ge 1\cdot|z+m-2|\cdot|z+m-1|\cdot 1\cdot 2\cdots (n-m)\nonumber\\
&\ge |z+m-2|\cdot|z+m-1|\cdot\frac{n!}{n^m}.\label{eq:order_2}
\end{align}
Next, since $n\ge m\ge 2$ and $m<|z|+2$, we have $n^m\le (2m)^n < (2|z|+4)^n$. Moreover, if we take $r$ in \eqref{eq:def_order} large enough so that $r>4$, then we want to prove \cref{eq:def_order} only for those $z$ with $|z|>4$. In this case, we can use $2|z|+4\le 3|z|$. It follows that $n^m < (2|z|+4)^n\le (3|z|)^n$.
Combining with \eqref{eq:order_1} and \eqref{eq:order_2} then leads to
\[
|a_n| \le \frac{1}{|z+m-2|\cdot|z+m-1|}\cdot\frac{(3\lceil \gamma \rceil \eta|z|)^n}{n!},
\]
for all $n\ge m$.
When $n \leq m-1$, then assuming $|z|>4$, we can bound
\[ \frac{1}{\left| z^{\uparrow n} \right|} \leq \frac{1}{(m-2)\cdots(m-1-n)|z-n+1|} \leq \begin{cases} \frac{|z|^n}{n!} &\text{ for } n \leq m-2,\\ \frac{|z|^n}{n!}\frac{1}{|z+m-2|} &\text{ for } n = m-1. \end{cases} \]
Assume first that $z$ is not close to $(2-m)$ or to $(1-m)$; for concreteness, say e.g.~that $\text{dist}(z,\Z_{\le 0})\ge\frac{1}{10}$, so that $|z+m-2|\cdot|z+m-1|\ge\frac{1}{100}$. Then our analysis results into
\[
\big| {}_1F_1\left(\gamma;z;-\eta\right) \big| \le \sum_{n\ge 0}{|a_n|} \le C\sum_{n\ge 0}{\frac{(C'|z|)^n}{n!}} = C\cdot\exp(C'|z|),
\]
for some constants $C,C'>0$.
From the asymptotic expansion of the inverse gamma function, we obtain $\left|\Gamma(z)^{-1}\right|\le \exp\left(C''|z|^2\right)$, for some constant $C''>0$, as long as $|z|$ is large enough.
As a result, $\big| F^{\text{planch}}(z) \big| = \left|\Gamma(z)^{-1}\right|\cdot \big| {}_1F_1\left(\gamma;z;-\eta\right) \big| \le C\cdot\exp\big(C'|z|+C''|z|^2\big)$, and therefore \cref{eq:def_order} follows, e.g.~when $A=3$, if $r>0$ is large enough.

If $z$ is close to $(2-m)$, that is, $|z+m-2|<\frac{1}{10}$, then automatically $|z+m-1|>\frac{9}{10}$ and this allows us to at least bound $\frac{z^{\uparrow n}}{(z+m-2)}$ as above. As a result,
\[
\big|(z+m-2)\cdot{}_1F_1(\gamma;z;-\eta)\big| \le C\cdot\exp(C'|z|),
\]
for some constants $C,C'>0$.
We can also bound $\left|(z+m-2)^{-1}\cdot\Gamma(z)^{-1}\right|\le\exp\left(C''|z|^2\right)$, for some $C''>0$, by Cauchy's integral formula.
Combining these two inequalities again gives the desired \cref{eq:def_order} for $A=3$ and large enough $r>0$.

If $|z+m-1|<\frac{1}{10}$, the proof is the same.

\subsection{The alpha case}

Let $\gamma,\eta>0$ and $c\in(0,1)$ be arbitrary.
It suffices to show that for
\[
\widetilde{F}^{\text{{alpha}}}(z) = \frac{1}{\Gamma(z)}\cdot{}_2F_1\left(\gamma,z-\eta;z;c\right) = \frac{1}{\Gamma(z)}\sum_{n\ge 0}{ \frac{\gamma^{\uparrow n}(z-\eta)^{\uparrow n}}{z^{\uparrow n}}\frac{c^n}{n!} },
\]
there exist $A,r>0$ such that
\begin{equation}\label{eq:alpha_finite_order}
\big| \widetilde{F}^{\text{alpha}}(z) \big| \le \exp\big( |z|^A \big),\text{ for all }|z|>r,
\end{equation}
and similarly as before, it is enough to prove it with an additional assumption that $|z|\notin\Z_{\ge 0}$.
The proof is similar to that for the Plancherel case, so we only elaborate on how to modify the main step of the proof, which is to bound the modulus of $(z-\eta)^{\uparrow n}/z^{\uparrow n}$.

The easier case is when $z$ is not close to the negative $x$-axis, in the sense that $|\arg{z}|\le\frac{9\pi}{10}$.
Then one can verify that $|\arg(z+n)|\le\frac{9\pi}{10}$ and $\frac{|z|}{|z+n|}\le\frac{1}{\sin(\pi/10)}=:C$, for all $n\in\Z_{\ge 0}$.
Using the known estimate for Gamma functions,
\[
\frac{\Gamma(x+a)}{\Gamma(x+b)} = x^{a-b}\big( 1+O(|x|) \big), \text{ as }|x|\to\infty,\ |\arg{x}|\le \pi-\delta<\pi,
\]
we get
\begin{equation}\label{eq:ineq_frac_2f1}
\left| \frac{(z-\eta)^{\uparrow n}}{z^{\uparrow n}} \right| = \left| \frac{\Gamma(z-\eta+n)\Gamma(z)}{\Gamma(z-\eta)\Gamma(z+n)} \right|
\lesssim \left(\frac{|z|}{|z+n|}\right)^\eta \le C^\eta,
\end{equation}
for all $n\in\Z_{\ge 0}$, in the region $|\arg{z}|\le\frac{9\pi}{10}$.
Above and throughout this proof, we use the notation $A(z,n)\lesssim B(z,n)$ if there exists a constant $C>0$, independent of $z,n$, such that $A(z,n)\le C\cdot B(z,n)$.
As a result,
\[
\big|{}_2F_1\left(\gamma,z-\eta;z;c\right)\big| \le \sum_{n\ge 0}{ \left|\frac{(z-\eta)^{\uparrow n}}{z^{\uparrow n}}\right| \cdot\frac{\gamma^{\uparrow n}c^n}{n!} }\lesssim \sum_{n\ge 0}{\frac{\gamma^{\uparrow n}c^n}{n!}} = (1-c)^{-\gamma}.
\]
Together with the bound $\left|\Gamma(z)^{-1}\right|\le \exp\left(C''|z|^2\right)$, already used in the Plancherel case, we obtain the desired \cref{eq:alpha_finite_order} for $A=3$ and large enough $r>0$.

The case $|\arg z|\in(\frac{9\pi}{10},\pi]$ remains.
Note that $\Re z<0$ in this case.
Let $m\in\Z_{\ge 1}$ be such that $m-1< -\Re z<m$.
The main step is to bound the modulus of $\frac{(z-\eta)^{\uparrow n}}{z^{\uparrow n}}$; we show how to do it when $n\ge m$.
Note that $\Re(z+m)=\Re z+m>0$, so $|\arg(z+m)|\in [0,\frac{\pi}{2}]$, and we can use the inequality \eqref{eq:ineq_frac_2f1} that we already proved.
Moreover, $m-1< -\Re z=|\Re z|\le|z|$ implies that $|z|,|z+1|,\dots,|z+m-2|\ge 1$.
Combining all these inequalities, we deduce
\begin{align*}
\left| \frac{(z-\eta)^{\uparrow n}}{z^{\uparrow n}} \right| &= \left| \frac{(z-\eta)\cdots (z-\eta+m-1)(z-\eta+m)\cdots (z-\eta+n-1)}{z(z+1)\cdots (z+m-1)(z+m)\cdots (z+n-1)} \right|\\
&\le\frac{(|z|+\eta)\cdots (|z|+\eta+m-1)}{|z|\cdot|z+1|\cdots|z+m-1|}\cdot\left|\frac{(z+m-\eta)^{\uparrow (n-m)}}{(z+m)^{\uparrow (n-m)}}\right|\\
&\lesssim \frac{(|z|+\eta)\cdots (|z|+\eta+m-1)}{|z|\cdot|z+1|\cdots|z+m-1|}
\le\frac{(|z|+\eta+m-1)^m}{|z+m-1|}
\le\frac{(2|z|+\eta)^{|z|+1}}{|z+m-1|}.
\end{align*}
Let us take $r$ in \cref{eq:alpha_finite_order} to satisfy $r>\max\{\eta,4\}$, so that we can use $\eta<|z|$; also, $4<|z|$ implies $3|z|\le e^{|z|-1}$.
As a result,
\[
\left| \frac{(z-\eta)^{\uparrow n}}{z^{\uparrow n}} \right|
\lesssim\frac{(2|z|+\eta)^{|z|+1}}{|z+m-1|} \le\frac{(3|z|)^{|z|+1}}{|z+m-1|} \le\frac{(e^{|z|-1})^{|z|+1}}{|z+m-1|}
\le\frac{e^{|z|^2}}{|z+m-1|}.
\]
If $|z+m-1|\ge\frac{1}{10}$, then $\left|\frac{(z-\eta)^{\uparrow n}}{z^{\uparrow n}}\right|\lesssim\exp(|z|^2)$, for all $n\ge m$.
Some additional work is needed in the case $n<m$, but the same holds, therefore $\big|{}_2F_1\left(\gamma,z-\eta;z;c\right)\big|\le\sum_{n\ge 0}{ \left|\frac{(z-\eta)^{\uparrow n}}{z^{\uparrow n}}\right| \cdot\frac{\gamma^{\uparrow n}c^n}{n!} }\lesssim\exp(|z|^2)$. Together with $\left|\Gamma(z)^{-1}\right|\le \exp\left(C''|z|^2\right)$, we obtain the desired \cref{eq:alpha_finite_order} for $A=3$ and large enough $r>0$.

If $|z+m-1|<\frac{1}{10}$, then we can prove that $|(z+m-1)\cdot{}_2F_1\left(\gamma,z-\eta;z;c\right)|\lesssim\exp(|z|^2)$ as above, and $\big| (z+m-1)^{-1}\cdot\Gamma(z)^{-1} \big|\le\exp\big(C''|z|^2\big)$ by Cauchy's integral formula, as we did in the Plancherel case. The desired \cref{eq:alpha_finite_order} follows and the proof is finished.

\bibliographystyle{amsalpha}

\bibliography{biblio2015}

\end{document}